\documentclass[aps,superscriptaddress,10pt,nopacs,twocolumn]{revtex4-2}
\usepackage[dvipsnames]{xcolor}
\usepackage{algorithm,algorithmicx,algpseudocode}

\usepackage{amsfonts,amsmath,amssymb,mathtools,mathrsfs,bbm,float}
\usepackage{graphicx,epic,eepic,epsfig,latexsym,verbatim,color}
\usepackage[inline,shortlabels]{enumitem}
\usepackage[caption=false]{subfig}
 
\usepackage{tikz}
\usetikzlibrary{chains}
\usetikzlibrary{fit}
\usepackage{pgflibraryarrows}		%optional
\usepackage{pgflibrarysnakes}		%optional

\usepackage{epsfig}
\usetikzlibrary{shapes.symbols,patterns} % for source symbols
\usepackage{pgfplots}

\usepackage[strict]{changepage}

\usepackage[marginal]{footmisc}
\usepackage{url}
\usepackage{theorem}
\usepackage{thmtools}
\usepackage[ISO]{diffcoeff}
\usepackage{nicematrix}
\usepackage{xfrac}

\newtheorem{definition}{Definition}
\newtheorem{proposition}{Proposition}
\newtheorem{lemma}[proposition]{Lemma}

\newtheorem{theorem}[proposition]{Theorem}

\newtheorem{remark}{Remark}
\newenvironment{proof}{\noindent \textbf{{Proof~} }}{\hfill $\blacksquare$}
\def\squareforqed{\hbox{\rlap{$\sqcap$}$\sqcup$}}
\def\qed{\ifmmode\squareforqed\else{\unskip\nobreak\hfil
\penalty50\hskip1em\null\nobreak\hfil\squareforqed
\parfillskip=0pt\finalhyphendemerits=0\endgraf}\fi}
\def\endenv{\ifmmode\;\else{\unskip\nobreak\hfil
\penalty50\hskip1em\null\nobreak\hfil\;
\parfillskip=0pt\finalhyphendemerits=0\endgraf}\fi}

\newcounter{example}

\mathchardef\ordinarycolon\mathcode`\:
\mathcode`\:=\string"8000
\def\vcentcolon{\mathrel{\mathop\ordinarycolon}}
\begingroup \catcode`\:=\active
  \lowercase{\endgroup
  \let :\vcentcolon
  }

\usepackage{hyperref}
\hypersetup{colorlinks=true,citecolor=blue,linkcolor=blue,filecolor=blue,urlcolor=blue,breaklinks=true}
\usepackage{footnotebackref}
\usepackage[capitalise]{cleveref}
\usepackage{graphicx}
\usepackage{xcolor}

\definecolor{darkblue}{RGB}{0,76,156}
\definecolor{darkkblue}{RGB}{0,0,153}
\definecolor{blue2}{RGB}{102,178,255}
\definecolor{darkred}{RGB}{195,0,0}

\RequirePackage[framemethod=default]{mdframed}
\newmdenv[skipabove=7pt,
skipbelow=7pt,
backgroundcolor=darkblue!15,
innerleftmargin=5pt,
innerrightmargin=5pt,
innertopmargin=5pt,
leftmargin=0cm,
rightmargin=0cm,
innerbottommargin=5pt,
linewidth=1pt]{tBox}

\newmdenv[skipabove=7pt,
skipbelow=7pt,
backgroundcolor=blue2!25,
innerleftmargin=5pt,
innerrightmargin=5pt,
innertopmargin=5pt,
leftmargin=0cm,
rightmargin=0cm,
innerbottommargin=5pt,
linewidth=1pt]{dBox}

\newmdenv[skipabove=7pt,
skipbelow=7pt,
backgroundcolor=darkred!15,
innerleftmargin=5pt,
innerrightmargin=5pt,
innertopmargin=5pt,
leftmargin=0cm,
rightmargin=0cm,
innerbottommargin=5pt,
linewidth=1pt]{rBox}

\newcommand{\nc}{\newcommand}
\nc{\ketbra}[2]{\lvert#1\rangle\!\langle#2\rvert}
\DeclarePairedDelimiter{\norm}{\lVert}{\rVert}
\DeclarePairedDelimiter{\abs}{\lvert}{\rvert}

\makeatletter
\let\oldabs\abs
\def\abs{\@ifstar{\oldabs}{\oldabs*}}
\let\oldnorm\norm
\def\norm{\@ifstar{\oldnorm}{\oldnorm*}}
\makeatother

\DeclarePairedDelimiterX{\infdivx}[2]{(}{)}{%
  #1\;\delimsize\|\;#2%
}

\nc{\proj}[1]{| #1\rangle\!\langle #1 |}
\nc{\avg}[1]{\langle#1\rangle}
\nc{\smfrac}[2]{\mbox{$\frac{#1}{#2}$}}
\nc{\tr}{\operatorname{tr}}
\nc{\ox}{\otimes}
\nc{\dg}{\dagger}
\nc{\dn}{\downarrow}
\nc{\cA}{{\cal A}}
\nc{\cB}{{\cal B}}
\nc{\cC}{{\cal C}}
\nc{\cD}{{\cal D}}
\nc{\cE}{{\cal E}}
\nc{\cF}{{\cal F}}
\nc{\cG}{{\cal G}}
\nc{\cH}{{\cal H}}
\nc{\cI}{{\cal I}}
\nc{\cJ}{{\cal J}}
\nc{\cK}{{\cal K}}
\nc{\cL}{{\cal L}}
\nc{\cM}{{\cal M}}
\nc{\cN}{{\cal N}}
\nc{\cO}{{\cal O}}
\nc{\cP}{{\cal P}}
\nc{\cQ}{{\cal Q}}
\nc{\cR}{{\cal R}}
\nc{\cS}{{\cal S}}
\nc{\cT}{{\cal T}}
\nc{\cU}{{\cal U}}
\nc{\cV}{{\cal V}}
\nc{\cX}{{\cal X}}
\nc{\cY}{{\cal Y}}
\nc{\cZ}{{\cal Z}}
\nc{\cW}{{\cal W}}
\nc{\csupp}{{\operatorname{csupp}}}
\nc{\qsupp}{{\operatorname{qsupp}}}
\nc{\var}{{\operatorname{var}}}
\nc{\rar}{\rightarrow}
\nc{\lrar}{\longrightarrow}
\nc{\polylog}{{\operatorname{polylog}}}
\nc{\wt}{{\operatorname{wt}}}
\nc{\supp}{{\operatorname{supp}}}

\nc{\argmin}{{\operatorname{argmin}}}

\makeatletter
\newcommand{\tpmod}[1]{{\@displayfalse\pmod{#1}}}
\makeatother

\def\x{\xi}

\nc{\RR}{{{\mathbb R}}}
\nc{\CC}{{{\mathbb C}}}
\nc{\FF}{{{\mathbb F}}}
\nc{\NN}{{{\mathbb N}}}
\nc{\ZZ}{{{\mathbb Z}}}
\nc{\PP}{{{\mathbb P}}}
\nc{\QQ}{{{\mathbb Q}}}
\nc{\UU}{{{\mathbb U}}}
\nc{\EE}{{{\mathbb E}}}
\nc{\id}{{\operatorname{id}}}

\nc{\CHSH}{{\operatorname{CHSH}}}

\nc{\rU}{\mbox{U}}

\nc{\ob}[1]{#1}

\nc{\SEP}{{\text{\rm SEP}}}
\nc{\NS}{{\text{\rm NS}}}
\nc{\LOCC}{{\text{\rm LOCC}}}
\nc{\PPT}{{\text{\rm PPT}}}
\nc{\EXT}{{\text{\rm EXT}}}
\nc{\Sym}{{\operatorname{Sym}}}

\nc{\ERLO}{{E_{\text{r,LO}}}}
\nc{\ERLOCC}{{E_{\text{r,LOCC}}}}
\nc{\ERPPT}{{E_{\text{r,PPT}}}}
\nc{\ERLOCCinfty}{{E^{\infty}_{\text{r,LOCC}}}}
\nc{\Aram}{{\operatorname{\sf A}}}
\newtheorem{problem}{Problem}

\newcommand{\eps}{\varepsilon}

\makeatletter
\def\grd@save@target#1{%
  \def\grd@target{#1}}
\def\grd@save@start#1{%
  \def\grd@start{#1}}
\tikzset{
  grid with coordinates/.style={
    to path={%
      \pgfextra{%
        \edef\grd@@target{(\tikztotarget)}%
        \tikz@scan@one@point\grd@save@target\grd@@target\relax
        \edef\grd@@start{(\tikztostart)}%
        \tikz@scan@one@point\grd@save@start\grd@@start\relax
        \draw[minor help lines,magenta] (\tikztostart) grid (\tikztotarget);
        \draw[major help lines] (\tikztostart) grid (\tikztotarget);
        \grd@start
        \pgfmathsetmacro{\grd@xa}{\the\pgf@x/1cm}
        \pgfmathsetmacro{\grd@ya}{\the\pgf@y/1cm}
        \grd@target
        \pgfmathsetmacro{\grd@xb}{\the\pgf@x/1cm}
        \pgfmathsetmacro{\grd@yb}{\the\pgf@y/1cm}
        \pgfmathsetmacro{\grd@xc}{\grd@xa + \pgfkeysvalueof{/tikz/grid with coordinates/major step}}
        \pgfmathsetmacro{\grd@yc}{\grd@ya + \pgfkeysvalueof{/tikz/grid with coordinates/major step}}
        \foreach \x in {\grd@xa,\grd@xc,...,\grd@xb}
        \node[anchor=north] at (\x,\grd@ya) {\pgfmathprintnumber{\x}};
        \foreach \y in {\grd@ya,\grd@yc,...,\grd@yb}
        \node[anchor=east] at (\grd@xa,\y) {\pgfmathprintnumber{\y}};
      }
    }
  },
  minor help lines/.style={
    help lines,
    step=\pgfkeysvalueof{/tikz/grid with coordinates/minor step}
  },
  major help lines/.style={
    help lines,
    line width=\pgfkeysvalueof{/tikz/grid with coordinates/major line width},
    step=\pgfkeysvalueof{/tikz/grid with coordinates/major step}
  },
  grid with coordinates/.cd,
  minor step/.initial=.2,
  major step/.initial=1,
  major line width/.initial=2pt,
}
\makeatother

\usepackage{thmtools}
\usepackage{thm-restate}
\usepackage{etoolbox}
\makeatletter
\def\problem@s{}
\newcounter{problems@cnt}

\newcommand{\allproblems}{\problem@s}
\makeatother

\usepackage{amsmath,amsfonts,amssymb,array,graphicx,mathtools,multirow,bm,times,tcolorbox,relsize,booktabs}
\usepackage[utf8]{inputenc}
\usepackage[T1]{fontenc}
\usepackage{qcircuit}
\usepackage{makecell}
\usepackage{braket}
\usepackage{natbib}

\definecolor{colortwo}{rgb}{0.4,0.77,0.17}
\definecolor{colorthree}{rgb}{0.01,0.51,0.93}

\nc{\cnot}{\mathrm{CNOT}}
\nc{\paulinorm}[2]{\norm{#1}_{\mathrm{Pauli}, #2}}
\allowdisplaybreaks

\begin{document}
\title{Efficient information recovery from Pauli noise via classical shadow}
\author{Yifei Chen}
\author{Zhan Yu}
\author{Chenghong Zhu}
\author{Xin Wang}
\email{wangxin73@baidu.com}
\affiliation{Institute for Quantum Computing, Baidu Research, Beijing 100193, China}
\begin{abstract}
The rapid advancement of quantum computing has led to an extensive demand for effective techniques to extract classical information from quantum systems, particularly in fields like quantum machine learning and quantum chemistry. However, quantum systems are inherently susceptible to noises, which adversely corrupt the information encoded in quantum systems. In this work, we introduce an efficient algorithm that can recover information from quantum states under Pauli noise. The core idea is to learn the necessary information of the unknown Pauli channel by post-processing the classical shadows of the channel. For a local and bounded-degree observable, only partial knowledge of the channel is required rather than its complete classical description to recover the ideal information, resulting in a polynomial-time algorithm. This contrasts with conventional methods such as probabilistic error cancellation, which requires the full information of the channel and exhibits exponential scaling with the number of qubits. We also prove that this scalable method is optimal on the sample complexity and generalise the algorithm to the weight contracting channel. Furthermore, we demonstrate the validity of the algorithm on the 1D anisotropic Heisenberg-type model via numerical simulations. As a notable application, our method can be severed as a sample-efficient error mitigation scheme for Clifford circuits.
\end{abstract}

% Identifying the nature of the noise presents a challenge because it may correspond to a quantum process whose complete characterisation scales exponentially with the number of qubits. On top of that, the question remains as to how to effectively mitigate the impact of noise. 

\date{\today}
\maketitle

\section{Introduction}
Quantum computers are shown to be able to solve certain problems significantly faster than classical computers~\cite{Childs2010}. However, current quantum devices are susceptible to noise from various sources like the environment, crosstalk, and quantum decoherence, which sets an ultimate time and size limit for quantum computation. Thus the near-term state of quantum computing is referred to as the noisy intermediate-scale quantum (NISQ) era~\cite{preskill2018quantum}. To unleash the potential of near-term quantum computers, a major challenge is to reduce the effect of noise in the quantum system.

One of the most important ingredients in quantum computing is to extract information from a quantum system by measuring the quantum state, which is described as the expectation value of some observable~$O$ of interest. The expectation value of some chosen observable unravels many properties of the quantum system, which is extensively used in many quantum algorithms, including variational quantum eigensolver~\cite{peruzzo2014variational}, quantum approximate optimization algoirhtm~\cite{farhi2014quantum}, and quantum machine learning~\cite{biamonte2017quantum}.

For an ideal quantum state $\sigma$, the information that we seek to obtain is $\tr(O\sigma)$. However, due to the noise present in the quantum computer, the actual state in practice is some \emph{noisy state} $\tilde{\sigma}$ instead.  One of the most standard theoretical models for quantum noise in the study of quantum error correction and mitigation is \emph{Pauli noise}. On one hand, Pauli noise provides a simple model that describes common incoherent noise such as bit-flip, depolarizing, and dephasing. On the other hand, general quantum noise can be mapped to Pauli noise without incurring a loss of fidelity by the technique of randomised compiling~\cite{wallman2016noise, hashim2021randomized}.

The problem of recovering from Pauli noise is that, given access to an unknown Pauli noise $\cP$ and copies of the noisy state $\tilde{\sigma} = \cP(\sigma)$, retrieve the information $\tr(O\sigma)$ for some observable $O$. To recover from a noise $\cP$, a natural way is to construct a map $\cQ$ such that the composed map $\cQ \circ \cP$ is an identity map~\cite{jiang2021physical}, which could covert the noisy state $\cP(\sigma)$ to the ideal state $\sigma$. Such a map is actually not necessary if we are only concerned with the target expectation value $\tr(O\sigma)$ instead of the ideal state $\sigma$. \citet{zhao2023information} proved the necessary and sufficient condition for retrieving the target information from noisy quantum states, and utilised semidefinite programming to determine an optimal protocol for constructing the map $\cD$ that satisfies $\tr(\cD\circ \cP(\sigma) O) = \tr(O \sigma)$. While this method is not restricted to the class of Pauli channels, it requires complete information of the quantum noise. Obtaining the full classical description of an unknown Pauli channel often uses techniques like quantum process tomography~\cite{chuang1997prescription, altepeter2003ancillaassisted, mohseni2008quantumprocess}, typically requiring a number of copies of the channel that scales exponentially in the number of qubits, which is resource-consuming and inefficient. Furthermore, the map $\cD$ proposed in Ref.~\cite{zhao2023information} needs to be simulated via probabilistic sampling, which requires additional resources. How to efficiently recover information from a Pauli channel with no prior information still remains an open and challenging problem.

\begin{figure*}[htbp]
  \includegraphics[width=\textwidth]{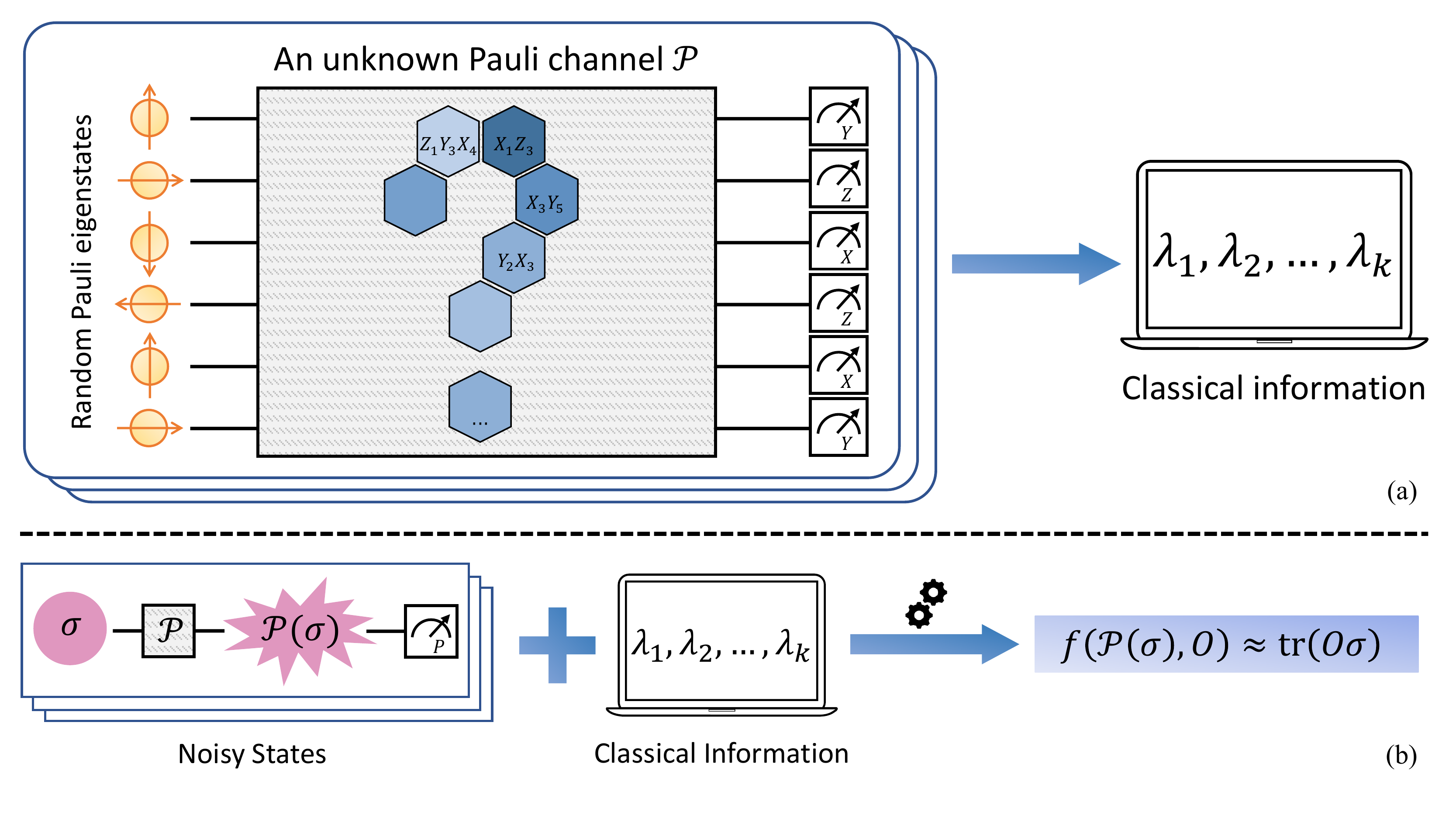}
  \caption{\textbf{Illustration of the algorithm for recovering information from Pauli noise.} (a) The classical information contains eigenvalues of the Pauli channel defined in \cref{subsec:channel prelim}. It is first estimated using the classical shadows of the channel, which is obtained by preparing random Pauli eigenstates as input and measuring the output states in random Pauli basis. (b) Then given any quantum states $\mathcal{P}(\sigma)$ that is subjected to this Pauli noise, estimation of $\tr(O\sigma)$ can be obtained by post-processing measurement results of the noisy state and the classical information we learnt. We note that the same classical information can be reused to recover information for different noisy states.
  }
  \label{fig:main}
\end{figure*}

To make progress towards resolving this open problem, we consider only obtaining partial information of the unknown Pauli channel instead of its full classical description, which would be sufficient for retrieving the expectation value of specific observables. We then note that there are efficient methods such as shadow tomography~\cite{aaronson2018shadow}, classical shadow~\cite{huang2020predicting} and recently proposed quantum estimation algorithms~\cite{huang2022learning, flammia2021pauli, chen2022quantum, caro2022learning} that can estimate these properties of a quantum system using very few quantum resources. This provides us the intuition that the technique of classical shadow tomography has the potential to lead to an efficient method of retrieving information from Pauli noise.

In this work, we propose an efficient algorithm that retrieves the information $\tr(O\sigma)$ from unknown Pauli noise $\cP$ for arbitrary $n$-qubit noisy state $\cP(\sigma)$ and \emph{bounded-degree $k$-local observables} $O$. The main idea is that, when the observable of interest is local and bounded-degree, then only partial eigenvalues of the channel are required to recover the ideal information. The algorithm consists of two steps: learning the necessary information of the unknown Pauli channel and using the information to estimate the expectation value $\tr(O\sigma)$ in classical post-processing. The main scheme is illustrated in \cref{fig:main}. For the learning process, we leverage the techniques of classical shadow tomography~\cite{huang2020predicting, huang2022learning} to estimate the eigenvalues of the Pauli channel up to precision $\epsilon$ with probability $1-\delta$, which only requires $\cO(\log(n^k/\delta)/\epsilon^2)$ copies of the Pauli channel. Furthermore, we utilise information-theoretic techniques~\cite{huang2021informationtheoretic, chen2022quantum} to prove a lower bound on the sample complexity, showing the optimality of our learning algorithm. We could apply classical shadow tomography on the noisy state $\tilde\sigma$ for bounded-degree observables to obtain necessary classical information, the sample complexity of which is also optimal~\cite{huang2022learning}. By post-processing the obtained information from these two steps, we retrieve the target expectation value $\tr(O\sigma)$ in computational time $\cO(n^k\log(n^k))$. 
As a notable application, we apply our method to mitigate Pauli errors in Clifford circuits, which leads to a more sample-efficient Pauli error mitigation scheme than previous methods such as probabilistic error cancellation~\cite{temme2017error}.

We will start by giving some background and introducing the idea of classical shadows in \cref{sec:prelim}, and then present the algorithm for recovering information from Pauli noise in \cref{sec:alg}. In \cref{sec:analysis}, we analyse the sample complexity and computational complexity of our proposed algorithm, which shows the efficiency and optimality of the algorithm. We present a numerical experiment in \cref{sec:numerical_experiments} showing the correctness of our algorithm. An application of the algorithm, which is error mitigation of Clifford circuits, is described in \cref{sec:application}. We discuss the comparison with prior work in \cref{sec:comparison} and conclude with outlook in \cref{sec:conclusion}.

\section{Preliminaries} \label{sec:prelim}
\subsection{Quantum channels and observables} \label{subsec:channel prelim}
In the theory of quantum information~\cite{Wilde2017book,Watrous2011b,Hayashi2017b}, noise of quantum systems are modelled by \emph{quantum channels}, which are completely positive and trace-preserving (CPTP) maps between spaces of operators. An $n$-qubit \emph{Pauli channel} is defined as
\begin{equation}
    \cP(\sigma) \coloneqq \sum_{P\in\{I, X, Y, Z\}^{\otimes n}} p(P) \cdot P \sigma P^\dagger,
\end{equation}
where $P$ is an $n$-fold tensor product of Pauli operators in $\{I, X, Y, Z\}$, and $p$ is a probability distribution on $\{I, X, Y, Z\}^{\otimes n}$. 
A quantum channel is \emph{unital} if it maps the identity operator to the identity operator. The adjoint map of an $n$-qubit quantum channel $\mathcal{N}$ is the unique map $\mathcal{N}^\dagger$ that satisfies 
\begin{equation}
    \tr ( X \mathcal{N}(Y)) = \tr (\mathcal{N}^\dagger(X) Y)
\end{equation}
for all linear operators $X,Y\in \mathcal{L}_{\mathbb{C}^{2n}}$ and $\mathcal{N}^\dagger$ is a completely positive and unital map. In particular,  $\mathcal{N}^\dagger$ maps hermitian operators to hermitian operators.

A Pauli channel $\mathcal{P}$ is in fact self-adjoint, meaning $\mathcal{P}^\dagger = \mathcal{P}$, which can be verified directly from the definition, so throughout this paper, we omit $\dagger$ on $\mathcal{P}^\dagger$. Another observation is that every $n$-qubit Pauli operator $P$ is an eigenoperator of $\mathcal{P}$ since 
\begin{align}
    \mathcal{P} (P) &= \sum_{Q\in \{ I,X,Y,Z\}^{\otimes n}} p(Q) QPQ \\
    &= \sum_{Q\in \{ I,X,Y,Z\}^{\otimes n}} (-1)^{\langle P,Q \rangle } p(Q) Q^2P \\
    &= \sum_{Q\in \{ I,X,Y,Z\}^{\otimes n}} (-1)^{\langle P,Q \rangle} p(Q) P, \\
    \text{where } \langle P,Q \rangle &= \begin{cases}
        0 & \text{if $P$ and $Q$ commute}, \\
        1 & \text{if $P$ and $Q$ anti-commute}.
    \end{cases}
\end{align}
The quantity $\sum_{Q\in \{ I,X,Y,Z\}^{\otimes n}} (-1)^{\langle P,Q \rangle} p(Q)$ is the eigenvalue of $P$ which we denote as $\lambda_P$. We refer to the collection of $\lambda_P$ as eigenvalues of the Pauli channel.

Observables are represented by hermitian operators. An observable $O$ is \emph{$k$-local} if it can be written as a linear combination $O = \sum_j \alpha_j O_j$ where each $O_j$ acts on at most $k$ qubits. An observable is \emph{bounded-degree} if only a constant number of terms $O_j$ in the sum act on each qubit.
The \emph{weight} of an $n$-qubit Pauli operator $P$, denoted as $\abs{P}$, is the number of tensor factors that are not identity $I$. Since Pauli operators form a basis of hermitian operators, any observable $O$ has a unique Pauli decomposition $O = \sum_P \alpha_P P$. 
This allows us to define the \emph{weight} of an observable $O$ to be the maximum weight of Pauli operators whose coefficient $\alpha_P$ is non-zero in the expansion of $O$. 
We also define the Pauli $p$-norm of an observable $O$, denoted as $\paulinorm{O}{p}$, to be the $l_p$-norm of $\bm\alpha$, where $\bm\alpha$ is the vector of Pauli coefficients $\alpha_P$.

\subsection{Classical shadow tomography}
In Ref.~\cite{aaronson2018shadow}, the author showed that for the task of estimating multiple measurement probability of an unknown state, only a sample size that is logarithmic in the number of measurements to predict and the dimension of the quantum state is required. Based on this work, \author{huang2020predicting}~\citet{huang2020predicting} considered the task of predicting $\tr(O\sigma)$ for a set of $O$ simultaneously under some mild conditions and the method proposed is called \emph{classical shadow tomography}. Classical shadows refer to the classical data acquired by performing randomised measurements on an unknown state. This can be realised by randomly selecting a unitary from a given set, applying it to the state and measuring the output state in the computational basis. It was shown that if the set of unitaries satisfies certain conditions, we can always construct an unbiased estimator for the state using classical shadows. There have been various recent progresses exploring applications and extensions of classical shadows, see, e.g., Refs.~\cite{Hadfield2020,Wu2021,Gebhart2023,Coopmans2023,Nguyen2022,Zhao2021a,Huang2022,Elben2022,Low2022,Bu2022,Wan2022,Becker2022}.

A common set of measurements is Pauli measurements. Its estimator is easy to compute and has the following performance guarantee:
\begin{proposition}[Theorem~1 and Proposition~3 in Ref.~\cite{huang2020predicting}] \label{prop:state_pauli_shadow}
Adopting a random Pauli basis primitive, where each random unitary is of the form $U_1\otimes\cdots\otimes U_n$, and each $U_i$ is uniformly selected from the single-qubit Clifford group. Given a collection of $k$-local observables $\{O_1,O_2,\ldots, O_M\}$, accuracy parameters $\epsilon,\delta \in [0,1]$,  then 
    \begin{align}
        N = \mathcal{O}\left( \frac{\log(M/\delta)}{\epsilon^2} \max_i 4^k\lVert O_i \rVert^2_\infty \right)
    \end{align}
    samples are required to simultaneously predict each $\tr(O_i \rho)$ up to accuracy $\epsilon$ with success probability $1-\delta$.
\end{proposition}
Here, adopting a random Pauli basis primitive means we are measuring each qubit in random Pauli basis. This is realised by applying a random single-qubit Clifford gate to each qubit and measuring in computational basis.
One can obtain this result by combining Theorem~1 and Proposition~3 in Ref.~\cite{huang2020predicting}. Theorem 1 states that the number of samples $N$ is in $\mathcal{O}(\log(M/\delta) / \epsilon^2)$ multiplied by a quantity that depends on the set of random unitaries. Proposition 3 further shows that this quantity for random Pauli measurement is upper bounded by $\max_i 4^k\lVert O_i \rVert^2_\infty$.

\section{Algorithm for Information Recovery} \label{sec:alg}
Firstly, we formally define the problem of information recovery from noisy quantum states. Given access to an unknown $n$-qubit Pauli channel $\mathcal{P}$ and a noisy state $\cP(\sigma)$, for a known bounded-degree $k$-local observable $O$, the task is to provide a function $f(\cP(\sigma), O)$ that approximates the ideal expectation value $\tr(O\sigma)$ within some precision $\epsilon$, i.e.,
\begin{equation}
    \abs{f(\cP(\sigma), O) - \tr(O\sigma)} \leq \epsilon.
\end{equation}
For the target expectation value, the action of the channel $\cP$ on state $\sigma$ can be viewed as its adjoint map acting on the observable $O$, 
\begin{equation}
    \tr(O\mathcal{P}(\sigma)) = \tr(\mathcal{P} (O) \sigma).
\end{equation}
Hence, an estimation for $\tr(O\sigma)$ can be obtained by calculating $\tr (\overleftarrow{O} \mathcal{P}(\sigma))$ for an observable $\overleftarrow{O}$ such that $\mathcal{P}(\overleftarrow{O}) = O$. We have that any Pauli operator $P$ is an eigenoperator of $\mathcal{P}$, i.e., $\mathcal{P}(P) = \lambda_P P$. 
This means that if we obtain the estimated value $\widehat{\lambda}_P$ of $\lambda_P$, we can construct $\overleftarrow{O}=\sum_P \overleftarrow{\alpha}_P P $ by simply taking the Pauli decomposition $O = \sum_P \alpha_P P$ and let $\overleftarrow{\alpha}_P = \alpha_P / \widehat{\lambda}_P$, so that
\begin{equation}
    \mathcal{P} (\overleftarrow{O}) = \sum_P \frac{\lambda_P}{\widehat{\lambda}_P} \alpha_P P \approx \sum_P \alpha_P P = O.
\end{equation}
If $O$ is $k$-local then $\alpha_P $ is zero for every $P$ that has weight greater than $k$.
Hence our estimate for $\tr (O \sigma)$ is given by 
\begin{equation}
    f( \mathcal{P}(\sigma), O ) = \tr( \overleftarrow{O} \mathcal{P}(\sigma)) =  \sum_{P:\abs{P}\leq k} \overleftarrow{\alpha}_P \tr ( P \mathcal{P} (\sigma)).
\end{equation}

We now formalise the concepts and propose the algorithm that can recover information from Pauli channels. The detailed procedure is given in \cref{alg:main}:

\begin{algorithm}[H]
\caption{Information recovery from Pauli noise}
\label{alg:main}
\begin{algorithmic}[1]
    \Require Access to an unknown $n$-qubit Pauli channel $\mathcal{P}$, a bounded-degree $k$-local observable $O=\sum_P \alpha_P P$, copies of unknown noisy state $\mathcal{P}(\sigma)$.
    \Ensure An estimation of $\tr (O\sigma)$.
    \State Prepare $N$ random $n$-fold product Pauli eigenstates $\{\rho_i = \bigotimes_{j=1}^n \ketbra{s_{ij}}{s_{ij}}\}_{i=1}^N$, where each $\ket{s_{ij}}$ is one of the six eigenstates of single-qubit Pauli operators.
    \State Apply the unknown channel $\mathcal{P}$ on the $N$ random states and perform random Pauli measurements on each qubit, obtaining data $\{ \bigotimes_{j=1}^n \ketbra{t_{ij}}{t_{ij}} \}_{i=1}^N$. 
    \State For each $n$-qubit Pauli operator $P$ with $|P| \leq k$, compute $\widehat{\lambda}_P$ by \cref{eqn:xp,eqn:lambdap}.
    \State For $O = \sum_{P: |P|\leq k} \alpha_P P$, let $\overleftarrow{\alpha}_P = \alpha_P / \widehat{\lambda}_P $ for $|P|\leq k$. 
    \State Perform Pauli measurement on the noisy state $\mathcal{P}(\sigma)$ for each Pauli operator $P$ with $\abs{P} \leq k$ and construct the estimation as \[f(\mathcal{P}(\sigma), O) = \sum_{P:|P|\leq k} \overleftarrow{\alpha}_P \tr(P \mathcal{P}(\sigma)).\]
\end{algorithmic}
\end{algorithm}
In steps~$1$ and $2$, we send random Pauli eigenstates into the unknown channel $\cP$ and measure the output state in random Pauli basis. The data acquired in the first two steps, which bear similarity to classical shadows of a quantum state, are indeed classical shadows of the quantum process $\cP$. Next, we use the classical shadows to estimate the eigenvalues of $\mathcal{P}$. Specifically, we calculate the estimated eigenvalues $\widehat{\lambda}_P$ as follows. Let $x_P =(1/3)^{\abs{P}} \lambda_P$, then by \cref{lem:pauli evalue est} below, we can construct an estimator for $x_P$ as
\begin{align}
\widehat{x}_P &= \frac{1}{N} \sum_{i=1}^{N} \tr \big(P \bigotimes_{j=1}^{n} (3 \ketbra{t_{ij}}{t_{ij}} - I) \big) \tr (P \rho_i)\nonumber\\
  &= \frac{1}{N} \sum_{i=1}^N \prod_{j=1}^n \tr(P_j (3 \ketbra{t_{ij}}{t_{ij}} - I)\tr(P_j \ketbra{s_{ij}}{s_{ij}}), \label{eqn:xp}
\end{align}
where $\bigotimes_{j=1}^n (3\ketbra{t_{ij}}{t_{ij}}-I)$ is the unbiased estimator of $\mathcal{P}(\rho_i)$ using classical shadows as presented in Refs.~\cite{huang2020predicting, huang2022learning}. Then we obtain an estimator of the eigenvalue as 
\begin{equation}
    \widehat{\lambda}_P = 3^{|P|} \widehat{x}_P, \label{eqn:lambdap}
\end{equation} 
since $\lambda_P=3^{|P|}x_P$. A special case is when $P=I$, the fact that $\mathcal{P}$ is unital implies $\lambda_I=1$ so there is no need to estimate its value. 

This estimator is obtained by an adapted version of Lemma 16 in Ref.~\cite{huang2022learning}, the full statement of which can be found in \cref{lem:estimate coeff}. The purpose of the original lemma is to extract a particular expansion coefficient of a general $O=\sum_P \alpha_P P$. We focus on the case $O=\mathcal{P} (P) = \lambda_P P$  to extract the eigenvalue $\lambda_P$.

\begin{lemma} \label{lem:pauli evalue est}
    Let $\mathcal{P}$ be an $n$-qubit Pauli channel with eigenvalues $\{\lambda_P\}_P$ so that $\mathcal{P} (P) = \lambda_P P$ for Pauli operators $P\in \{I,X,Y,Z\}^{\otimes n}$, $\mathcal{D}^0$ be the uniform distribution of $n$-fold product Pauli eigenstates. We have 
    \begin{equation}\label{eqn:lambda estimation}
    \underset{\rho\sim \mathcal{D}^0}{\mathbb{E}} \tr (P \mathcal{P}(\rho)) \tr (P\rho) = \bigg( \frac{1}{3} \bigg)^{|P|} \lambda_P.        
    \end{equation}
\end{lemma}
We can see that $\widehat{x}_P$ is simply an empirical estimation of the expectation value on the left hand side of \cref{eqn:lambda estimation}. This lemma is derived as an adapted version of Lemma 16 in Ref.~\cite{huang2022learning} and the detailed proof is provided in Appendix \ref{appendix:pauli est}.

In step~$4$, we divide each coefficient $\alpha_P$ by the corresponding estimated eigenvalue to construct the observable $\overleftarrow{O}$ that can achieve $\abs*{ \tr(\overleftarrow{O} \mathcal{P}(\sigma)) - \tr(O\sigma) }\leq \epsilon$. In step~$5$, we  obtain the estimation $\tr(\overleftarrow{O} \mathcal{P}(\sigma))$ by performing Pauli measurements on the noisy state $\mathcal{P}(\sigma)$. Note that steps~$1$ to $3$ do not require information about $O$ and $\sigma$ but a promise of $O$ being $k$-local, hence can be done beforehand. Whenever we are given a new observable $O$ and noisy state $\mathcal{P}(\sigma)$, we only need to re-apply steps~$4$ and $5$.

In step $5$ of the algorithm, we require the values of $\tr(P \mathcal{P}(\sigma))$ for all Pauli operators $P$ such that $\abs{P}\leq k$. We denote the total number of such $P$ to be $T(n,k)$, which satisfies $T(n,k) = \mathcal{O}(n^k)$. Note that it is a problem which can be solved using classical shadows of quantum state. By \cref{prop:state_pauli_shadow}, we can estimate each $\tr(P \mathcal{P}(\sigma))$ up to $\frac{\epsilon}{2T(n,k)}$ using 
\begin{equation}
    \mathcal{O}\left(\frac{T(n,k)^2 \log(T(n,k))}{\epsilon^2}\right) = \mathcal{O}\left(\frac{n^{2k}\log(n^k)}{\epsilon^2}\right)
\end{equation} 
copies of $\mathcal{P}(\sigma)$. This has been proven to be optimal for this prediction task~\cite{huang2022learning}. Since this task has been studied thoroughly and the complexity is polynomial in $n$, for simplicity, we assume that we have $\tr(P \mathcal{P}(\sigma))$ and we do not consider sample complexity for obtaining $\tr(P \mathcal{P}(\sigma))$ for the rest of the paper.

\section{Analysis of sample complexity}\label{sec:analysis}
Next, we analyse the sample complexity and computational complexity of our proposed algorithm.
\begin{theorem}\label{Thm:Main}
Given an unknown $n$-qubit Pauli channel $\mathcal{P}$, a noisy state $\cP(\sigma)$, and an $n$-qubit bounded-degree $k$-local observable $O$ where $k=\mathcal{O}(1)$ and $\lVert O \rVert = 1$. For $\epsilon, \delta > 0$, there exists an algorithm that uses $N = \mathcal{O}(\log (n^{k}/ \delta)/\epsilon^2)$ accesses to the channel to obtain a function $f(\mathcal{P}(\sigma), O)$ such that
\begin{equation}\label{Error bound}
\abs{f(\mathcal{P}(\sigma),O) - \tr(O\sigma)} \leq \epsilon 
\end{equation}
with probability at least $1-\delta$. The computation time is $\mathcal{O}(n^k \log (n^{k}))$.
\end{theorem}
 
This theorem provides a strong guarantee that both sample complexity and total computation time scale polynomially as the number of qubits increases, meaning that our algorithm is practical and scalable for large quantum systems. 

By using the formula in \cref{lem:pauli evalue est} and Hoeffding's inequality, the sample complexity has direct connection with how accurate we need to estimate the eigenvalues, which is denoted by $\epsilon'$. We then use a series of bounding to create connection between $\epsilon'$ and $\epsilon$ which translates to the final sample complexity. The detailed proof can be found in \cref{appendix:proof main}.

In the first part of the algorithm where we learn the eigenvalues of the channel, the number of channels that we use is in $\mathcal{O}(\log (n^k)/\epsilon^2)$. We now use information-theoretic techniques in Refs.~\cite{huang2021informationtheoretic,fawzi2023lower,chen2022quantum} to prove a lower bound on the sample complexity for achieving this channel learning task.
\begin{proposition}\label{prop:lower bound}
Let $\mathbb{K}$ denote the set of $n$-qubit Pauli operators whose weight is at most $k$, i.e., $\mathbb{K} \coloneqq \Set{ P \in \{I, X, Y, Z\}^{\otimes n} : \abs{P}\leq k}$. Given an unknown Pauli channel $\mathcal{P}$, if an algorithm can estimate the eigenvalue of every $P\in \mathbb{K}$ up to accuracy $\epsilon$ from $N$ access of $\mathcal{P}$, where arbitrary input state can be prepared to be sent through the channel and arbitrary POVM can be used to measure the output state during each access, then $N=\Omega(\log(n^k)/\epsilon^2)$.
\end{proposition}
The proof is given in \cref{appendix: optimality}. The lower bound obtained matches our upper bound, which shows the optimality of our algorithm when channel can only be used once in each access. On the other hand, if more quantum resources are available, such as ancilla qubits or even using multiple copies of channel at the same time, similar to what is proven in Ref.~\cite{chen2022quantum}, more efficient algorithms could be possible. 

\begin{remark}
\Cref{alg:main} involves the collection of Pauli shadows of the channel. In \cref{appendix: clifford shadow}, we further explore the utilisation of Clifford shadows and we show that Clifford shadows cannot provide any improvement for sample complexity under the assumption that the locality of the observable $k=\mathcal{O}(1)$.  
\end{remark}

\section{Numerical experiments} \label{sec:numerical_experiments}
Estimating the expectation value $\tr(O\sigma)$ has many applications in quantum information processing. For example, variational quantum eigensolvers~\cite{peruzzo2014variational} are proposed for estimating the ground state energy, which requires tuning the parameters to minimise the expectation value as a cost function. Subsequently, we use numerical simulations to demonstrate that our algorithm can estimate $\tr(\overleftarrow{O} \cP(\sigma)) \approx \tr(O\sigma)$. In our simulations, the 2-qubit product Pauli channel is utilised. We also consider the noise level in current quantum devices and choose the noise parameters of the channel as shown in \cref{tbl:pauli-channel-parameter}. Each row in the table corresponds to the noise parameters associated with each qubit. The target observable we choose is the 1D anisotropic Heisenberg-type Hamiltonian with $n$ sites.
When the periodic boundary condition is closed and the magnetic field is included, the observable can be expressed as
\begin{equation}
    O = \sum_{j=1}^{n-1}(J_x\sigma_j^x \sigma_{j+1}^x + J_y\sigma_j^y \sigma_{j+1}^y + J_z\sigma_j^z \sigma_{j+1}^z + h\sigma_j^z),
\end{equation}
where $\sigma_i^a$ represents Pauli-$a$ operator acting on the $i$\textsuperscript{th} qubit, $J_x$, $J_y$, $J_z$ are the spin coupling strengths and $h_z$ is the magnetic field applied along the $z$ direction. We randomly choose these coefficients to be $J_x=0.27$, $J_y=0.42$, $J_z=0.76$ and $h_z=0.6$, then normalise it to have $\lVert O\rVert_\infty = 1$.

Then the algorithm starts by collecting classical shadows of the channel and we choose to sample $[1,20]\times 10^4$ classical shadows for estimating eigenvalues, and divide $\alpha_P$ by the estimated eigenvalue. To estimate the accuracy with the increased number of classical shadows, we randomly generate 500 Haar-random $n$-qubit states $\{ \sigma_1, \ldots, \sigma_{500} \}$ and directly calculate the ideal expectation value with respect to $H$. We then send the state through the channel $\mathcal{P}$ and compute the mean absolute error (MAE) for not performing any post-processing and performing our information recovery algorithm, respectively:
\begin{align}
    \parbox{65pt}{MAE with no \\post-processing} &= \frac{1}{500} \sum_{i=1}^{500} \abs{ \tr(O \mathcal{P}(\sigma_i)) - \tr(O \sigma_i) }, \\
    \parbox{65pt}{MAE with\\ post-processing} &= \frac{1}{500} \sum_{i=1}^{500} \abs{ f(\mathcal{P}(\sigma_i), O) - \tr(O \sigma_i) }.
\end{align}
Finally, to make comparisons, we find the ratio between the two MAEs,
\begin{equation}
    r = \frac{\text{MAE with post-processing}}{\text{MAE with no post-processing}},
\end{equation}
to show how much our algorithm improves the estimation. For every $10^4$ classical shadows, we repeat the above procedures 10 times and average over $r$. The experiment results are shown in \cref{fig:n=2}. It can be seen clearly that the ratio $r$ decreases as the number of samples increases, which indicates more precise classical information we are extracting.

\begin{table}[ht]
\begin{tabular}{|c|c|c|c|c|}
\hline
qubit number & $p_I$ & $p_X$ & $p_Y$ & $p_Z$  \\
 \hline
1 & 0.75 & 0.10 & 0.10 & 0.05 \\
2 & 0.77 & 0.09 & 0.09 & 0.05 \\
\hline
\end{tabular}
\caption{Noise parameters for each qubit.}
\label{tbl:pauli-channel-parameter}
\end{table}

\begin{figure}[h!]
    \centering
    \includegraphics[width=\linewidth]{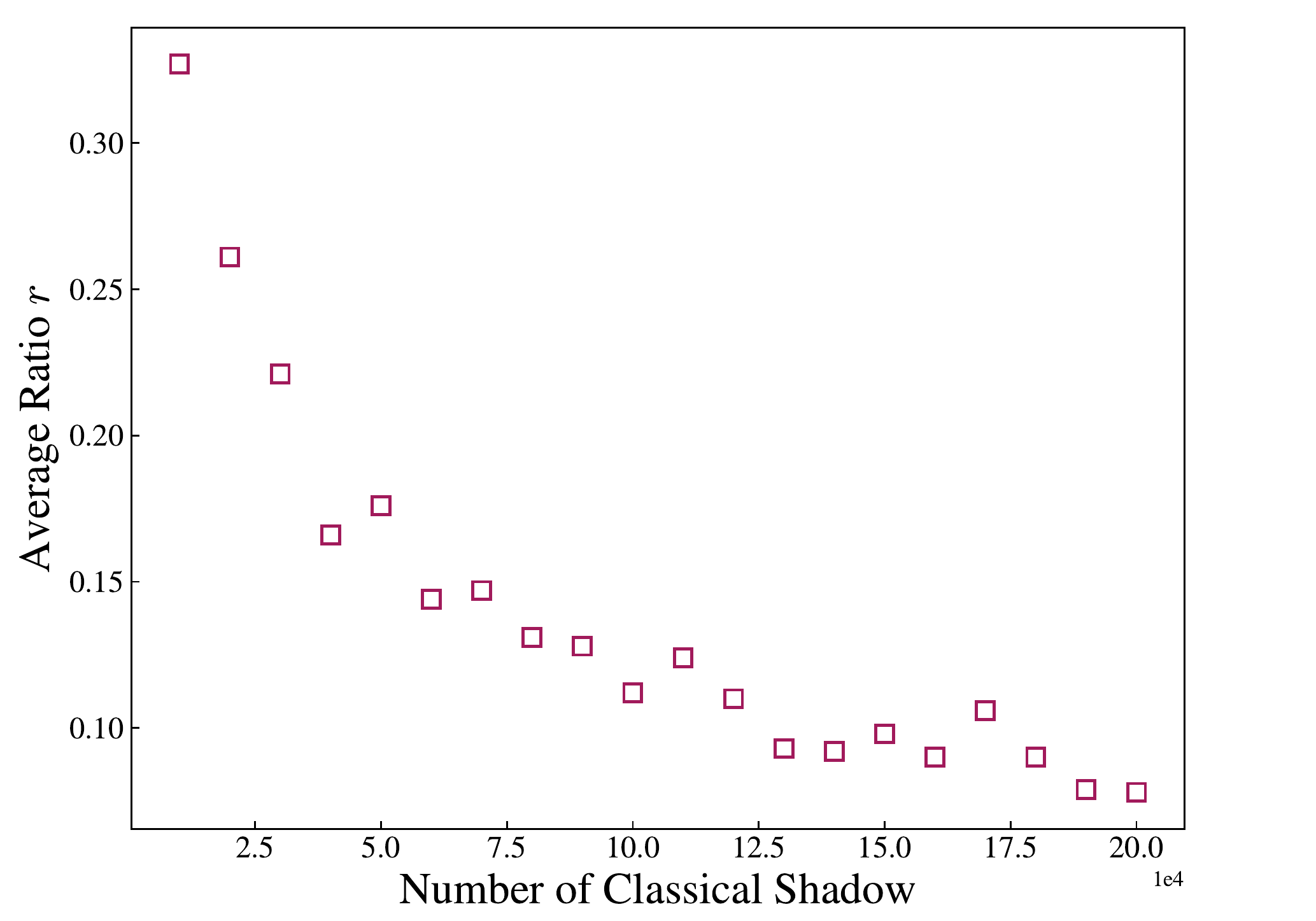}
    \caption{Information recovery on a 2-qubit anisoptropic Heisenberg-type model under product Pauli channels. The x-axis indicates the number of classical shadows we use, with a unit of $10^4$. The y-axis represents the ratio of error between when post-processing is performed and when no post-processing is performed. In comparison to extracting the expectation value from noisy states without post-processing, our algorithm provides a moderate increase in precision.}
    \label{fig:n=2} 
\end{figure}

Hence, we have demonstrated that our method can successfully predict the expectation value from noisy states and showed that the accuracy improves with more samples.

\section{Extensions in channel types}\label{sec:channel type}
In the previous section, we have demonstrated that our algorithm can recover information efficiently for Pauli channels. It is a natural question to ask what other types of channels lead to this efficiency.
A sufficient criterion is that the adjoint matrix of the channel written in Pauli basis is upper block triangular, where each block contains Pauli operators with the same weight and the block is arranged in increasing order of weight. An equivalent description is that the channel is \emph{weight contracting}, in the sense that the weight of the operator does not increase under the action of the adjoint map. One such example is product channel. In \cref{apx: extension}, we provide an updated version of the algorithm that can recover information from weight contracting channel which has the following updated performance guarantee:
\begin{proposition}\label{prop:extension}
Given an unknown $n$-qubit weight contracting channel $\mathcal{E}$, a noisy state $\cE(\sigma)$, and an $n$-qubit bounded-degree $k$-local observable O with $\norm{O}_\infty = 1$. For $\eps, \delta > 0$, there exists an algorithm that uses $N = \mathcal{O}(n^{2k}\log (n^{2k}/ \delta)/\epsilon^2)$ access to the channel to obtain a function $f$ such that
\begin{equation}
\lvert f(\mathcal{E}(\sigma)) - \tr(O\sigma) \rvert \leq \epsilon 
\end{equation}
with probability at least $1-\delta$. The computation time is $\mathcal{O}(n^{4k}\log(n^{2k}))$.
\end{proposition}
The proposition shows that the sample complexity of channels increases to $\mathcal{O}(n^{2k}\log(n^{2k}/\delta)/\epsilon^2)$ due to more information of the channel needs to be estimated, and the total computational time increases to $\mathcal{O}(n^{4k}\log(n^{2k}))$ accordingly, but still remains efficient and scalable.

\section{Application in Clifford circuit error mitigation} \label{sec:application}
An application of our method is mitigating Pauli errors in Clifford circuits, in which we only consider a circuit $\mathcal{C}$ consisting of $H$, $S$ and $\cnot$ gates.
Then, each gate is followed by a Pauli noise channel and we assume that noise for the same type of gate is the same. This is a stronger setting than the usual gate-independent time-stationary Markovian (GTM) noise considered in Refs.~\cite{chen2021robust, Flammia_2020, chen2022quantum} which assumes the noise channel after each gate is identical. We denote the resultant circuit as $\tilde{\mathcal{C}}$. For any input state $\sigma$, the expected output is $\mathcal{C}(\sigma)$ but the actual output is $\tilde{\mathcal{C}}(\sigma)$. The goal of error mitigation is then obtaining the exact expectation value $\tr (O \mathcal{C}(\sigma))$ for given observable $O$ when only the noisy version is available. 
To achieve this, we first learn the eigenvalues of the noise channel associated to the three types of gate. The method is the same as steps $1$ to $3$ of \cref{alg:main}, but the noisy channel is replaced by different noisy gates, and the corresponding estimators require slight modification. For single qubit gate, they become
\begin{align}
    \widehat{x}_P &= \frac{1}{N} \sum_{i=1}^N \tr( P ( 3\ketbra{t_{i}}{t_{i}} - I)) \tr (P (U\rho_i U^\dagger)), \\
    \widehat{\lambda}_P &= 3 \widehat{x}_P,
\end{align}
where $U\in \{ H,S \}$. For $\cnot$ gate, the matrix of noisy $\cnot$ gate is a monomial matrix. If we label the matrix entries by Pauli operators, then $(\pi(P), P)$ entry takes the value of $\lambda_P$ where $\pi(P)$ is the resultant Pauli operator by conjugating $P$ by $\cnot$. In this case, estimation is done by
\begin{align}
    \widehat{x}_P &= \frac{1}{N}  \sum_{i=1}^N \tr (P \bigotimes_{j=1}^n (3 \ketbra{t_{ij}}{t_{ij}} - I)) \tr (\pi(P) \rho_i), \\
    \widehat{\lambda}_P &= 3^{|\pi(P)|} \widehat{x}_P.
\end{align}
Next, we show how to error mitigate a Clifford circuit consisting of two gates and a higher number of gates follows the same idea. Let 
\begin{equation}
    \mathcal{C} = \mathcal{U}_2 \circ \mathcal{U}_1,
\end{equation} 
where $\mathcal{U}_i (\rho) = U_i \rho U_i^\dagger$, $U_i \in \{ H, S, \cnot \}$. The corresponding noisy circuit $\tilde{\mathcal{C}}$ is given by
\begin{equation}
    \tilde{\mathcal{C}} = \mathcal{P}_{\mathcal{U}_2} \circ \mathcal{U}_2 \circ \mathcal{P}_{\mathcal{U}_1} \circ \mathcal{U}_1
\end{equation}
Let $\lambda_{i,P}$ denotes the eigenvalue of $P$ under $\mathcal{P}_{\mathcal{U}_i}$, i.e. the eigenvalue of $P$ under the noise channel coupled with the $i$\textsuperscript{th} gate and $\pi_i(P) = U^\dagger_i P U_i$ be the Pauli operator we get when conjugate $P$ by $U_i^\dagger$. It can be shown that 
\begin{equation}
    \tr (P \tilde{\mathcal{C}}(\rho)) = \lambda_{2,P} \lambda_{1,\pi_2(P)} \tr( P \mathcal{C}(\rho)),
\end{equation}
hence we can construct 
\begin{equation}
    \overleftarrow{\alpha}_P = \frac{\alpha_P}{\widehat{\lambda}_{2,P} \widehat{\lambda}_{1,\pi_2(P)}},
\end{equation}
which would have $\tr (\overleftarrow{O} \tilde{\mathcal{C}}(\rho)) \approx \tr (O \mathcal{C}(\rho))$. 

\section{Comparison with prior work} \label{sec:comparison}
\textit{Comparison with prior methods of information recovery and error mitigation.}
Existing information recovery method~\cite{zhao2023information} is not limited to Pauli channels but it necessitates the full information of the channel is known. Prior to applying this method of information recovery, if the channel is unknown, it is necessary to obtain the full description of the channel via tomography, which requires an extensive amount of quantum resources. Another limitation of the method is that the effect of inaccurate channel description on the error of information recovery has not been adequately investigated, hence the sample complexity under approximate channel description is theoretically incomplete. Compared to the method in Ref.~\cite{zhao2023information}, our proposed algorithm of information recovery requires zero knowledge of a Pauli channel, and the computational complexity scales polynomially with the number of qubits.
Moreover, to implement the method outlined in Ref.~\cite{zhao2023information}, one needs to implement arbitrary CPTP maps, whereas our method only requires preparing Pauli eigenstates and performing Pauli measurements, both of which are considerably easier to realise on a quantum device.

There are also some methods used in quantum error mitigation that aim to obtain an estimate of noiseless information $\tr(O\sigma)$ using only copies of the noisy state. A commonly used method in error mitigation is probabilistic error cancellation (PEC)~\cite{temme2017error}, which starts by decomposing a target process as a linear combination of implementable noisy processes. Using this decomposition, the ideal circuit is realised by probabilistic sampling of noisy processes. By definition, it works for arbitrary quantum processes, but still faces similar problem of requiring full knowledge of the noisy processes. Another example is virtual distillation~\cite{koczor2021exponential, huggins2021virtual}, which assumes the dominant pure eigenvector of the mixed noisy state is the noiseless state. By using multiple copies of the noisy states, we can obtain the noiseless expectation value. However, copies need to be used at the same time which means the circuit width is high. Although there are some variants that trade circuit depth with width~\cite{czarnik2021qubitefficient} or combine with the framework of classical shadows~\cite{seif2023shadow} to reduce the circuit width, the total complexity is still exponential in the number of qubits.

\textit{Differences with learning to predict quantum processes.}
Our proposed algorithm is inspired by the method proposed in Ref.~\cite{huang2022learning} that aims to predict the value of $\tr(O \mathcal{P}(\rho))$ from access to $\mathcal{P}$ and $\rho$, whereas we want to recover the original information from noisy $\mathcal{P}(\rho)$. The authors do so by estimating the resultant observable $\mathcal{P}(O)$ from collected classical shadows. We use similar estimation techniques to learn the eigenvalues of the Pauli channel. In the end, they can guarantee accuracy to $\epsilon$ for the mean squared error over a restricted set of states but allow for arbitrary observable to be predicted. Whereas we can predict up to accuracy for the absolute error for any noisy state but only local observables. 

\textit{Comparison with Pauli channel learning.}
Not requiring a complete description of the noise channel is one of the main advantages of our algorithm over existing methods. We remark that numerous algorithms are capable of estimating the probability distribution $p(P)$ or the error rates for Pauli channels, which is commonly referred to as Pauli channel learning~\cite{chen2022quantum, Flammia_2021, Flammia_2020, fawzi2023lower, harper2020efficient}.
We make careful comparison with Ref.~\cite{Flammia_2021} whose setting is the most similar to the learning part of our algorithm.
In Ref.~\cite{Flammia_2021}, their main method involves preparing random product Pauli eigenstates whose eigenvalue is $+1$, send them through the channel then measure each qubit in the same basis as the input state. Using the data collected, they turn the problem of estimating error rates into a population recovery-type problem and use tools from that area to estimate single Pauli error rate. They do this for $\frac{4}{\epsilon}$ of the Pauli operators whose estimates are greater than $\frac{\epsilon}{2}$ and set the rest to $0$ to guarantee the efficiency of the algorithm while making sure that the $l_\infty$-norm of the difference is less than $\epsilon$. The sample complexity is $\mathcal{O}(\log(n/\epsilon\delta)/\epsilon^2)$.
Although the data collected is similar to our method, the post-processing and the quantity to estimate are quite different. Firstly, we also measure the output state in random Pauli basis. This allows us to extend our framework to a broader group of channel like product channels, which we have discussed in \cref{sec:channel type}. Secondly, the Pauli error rates are related to the Pauli eigenvalues by a Hadamard transform, i.e.
\begin{align}
    \lambda_P &= \sum_{Q \in \{I.X.Y.Z\}^{\otimes n}} (-1)^{\langle P,Q \rangle} p(Q) \\ \text{where } \langle P,Q \rangle &= \begin{cases}
        0 & \text{if $P$ and $Q$ commute} \\
        1 & \text{if $P$ and $Q$ anti-commute}
    \end{cases}
\end{align}
This means that obtaining the eigenvalue of a single Pauli operator requires knowledge of all Pauli error rates and vice versa. Also estimating error rates element wise to accuracy $\epsilon$ cannot guarantee to estimate eigenvalues element wise to accuracy $\epsilon$ and vice versa. 
The authors also present a method for estimating single eigenvalue for a given Pauli operator with sample complexity $\mathcal{O}(\log(1/\delta)/\epsilon^2)$, then treat it as a query access and use it to estimate the error rates, which is similar in spirit but their estimate is restricted to the specific Pauli that they query.

\section{Concluding remarks} \label{sec:conclusion}
In this work, we have introduced an efficient quantum algorithm that could retrieve information from an unknown Pauli noise by learning the channel and the noisy state, using quantum resources that scale polynomially in the number of qubits. The efficiency of the proposed algorithm comes from the fact that only partial knowledge of the channel is required to recover the ideal information for a local and bounded-degree observable. For learning partial eigenvalues of the Pauli channel, we have proved a lower bound on the sample complexity that matches the upper bound of our channel learning algorithm, which implies the optimality of the algorithm. We have also shown that the method can be directly applied to recover information from noisy Clifford circuits in a more efficient way than that of previous error mitigation methods such as probabilistic error cancellation. 

The method described in our work should be broadly applicable to mitigate Pauli noise for large-scale quantum devices. For further theoretical exploration, it would be interesting to extend the algorithm for information recovery to a wider range of quantum channels. For the practical aspect of this method, it is worthwhile to investigate the integration of general error mitigation schemes, which could lead to a potential resource-efficient method for early fault-tolerant quantum computers~\cite{Suzuki2022, Piveteau2021}.
We also anticipate that this method will be useful for reducing the effect of noise and improving the accuracy in near-term experiments. Specifically, we expect it to be applied for enhancing the performance of variational quantum algorithms on noisy devices with limited number of resources and contributing for exploring physically relevant properties in material science~\cite{ma2020quantum} and chemistry~\cite{nam2020ground,Cao2019}.

\begin{acknowledgements}
Part of this work was done when Y. C., Z. Y., and C. Z. were research interns at Baidu Research. X. W. would like to thank Xuanqiang Zhao for helpful discussions.
\end{acknowledgements}

\bibliography{ref}

\clearpage

\appendix
\vspace{2cm}
\onecolumngrid

% \begin{center}
% \textbf{
% {\large{Supplementary Material}}}
% \end{center}

\renewcommand{\theequation}{S\arabic{equation}}
\renewcommand{\theproposition}{S\arabic{proposition}}
\renewcommand{\thelemma}{S\arabic{lemma}}
\renewcommand{\thedefinition}{S\arabic{definition}}
\renewcommand{\thecorollary}{S\arabic{corollary}}
\renewcommand{\thefigure}{S\arabic{figure}}
\renewcommand{\theproblem}{S\arabic{problem}}
\setcounter{equation}{0}
\setcounter{table}{0}
\setcounter{section}{0}
\setcounter{proposition}{0}
\setcounter{lemma}{0}
\setcounter{corollary}{0}
\setcounter{figure}{0}
\setcounter{definition}{0}

\numberwithin{equation}{section}

\section{Proof of \cref{lem:pauli evalue est}} \label{appendix:pauli est}

\renewcommand\thelemma{\ref{lem:pauli evalue est}}
\setcounter{lemma}{\arabic{lemma}-1}
\begin{lemma}
    Let $\mathcal{P}$ be an $n$-qubit Pauli channel with eigenvalues $\{\lambda_P\}_P$ so that $\mathcal{P} (P) = \lambda_P P$ for Pauli operators $P\in \{I,X,Y,Z\}^{\otimes n}$, $\mathcal{D}^0$ be the uniform distribution of product state of Pauli eigenstates. We have 
    \begin{equation}
    \underset{\rho\sim \mathcal{D}^0}{\mathbb{E}} \tr (P \mathcal{P}(\rho)) \tr (P\rho) = \bigg( \frac{1}{3} \bigg)^{|P|} \lambda_P.        
    \end{equation}
\end{lemma}
\renewcommand{\thelemma}{S\arabic{lemma}}

\begin{proof}
We have
\begin{align}
    &\underset{\rho \sim \mathcal{D}^0}{\mathbb{E}} \tr(P\mathcal{P}(\rho)) \tr(P\rho) \\
    = &\underset{\rho \sim \mathcal{D}^0}{\mathbb{E}} \tr (\mathcal{P}(P) \rho) \tr(P \rho) \\
    = &\lambda_P \underset{\rho \sim \mathcal{D}^0}{\mathbb{E}} \tr(P\rho)^2.
\end{align}
Writing $\tr(P\rho) = \prod_i \tr(P_i \rho_i)$, we can see that if $P_i = I$, then $\tr(P_i \rho_i) = 1$. If $P_i \neq I$, then $\tr (P_i\rho_i)^2 = 1$ with probability $1/3$ and $\tr (P_i\rho_i)^2 = 0$ with probability $2/3$, hence $\underset{\rho \sim \mathcal{D}^0}{\mathbb{E}} \tr(P\rho)^2 = (1/3)^{|P|}$ and complete the proof.
\end{proof}

\section{Proof of \cref{Thm:Main}}\label{appendix:proof main}
Before proving the proposition, we first introduce a few lemmas that help us to bound the error.
\begin{lemma}\label{Difference}
The difference between two observable expectation estimations can be upper bounded by the difference in the Pauli decomposition of observables, 
\begin{equation}
   \lvert \tr (O_1 \rho) - \tr(O_2 \rho)  \rvert  \leq \sum_P \lvert \Delta \alpha_P \rvert ,
\end{equation}
where $\Delta \alpha_P$ is the coefficient of $P$ in the Pauli expansion of $O_1 - O_2$.
\end{lemma}
\begin{proof}
    \begin{align}
        &\lvert \tr (O_1 \rho) - \tr(O_2 \rho)  \rvert \\
        \leq &\lvert \sum_P \Delta \alpha_P \tr(P\rho) \rvert \\
        \leq & \sum_P \vert \Delta \alpha_P \rvert,
    \end{align}
where in the last line we use the triangle inequality and the fact that $\abs{\tr (P\rho)} \leq 1$.
\end{proof}\\
\begin{definition}
    The \emph{degree} of a bounded-degree observable is the maximum number of terms in the sum that act on each qubit.
\end{definition}
\begin{lemma}[Corollary 12 in Ref.~\cite{huang2022learning}]\label{lem:Pauli norm}
    Given an $n$-qubit $k$-local bounded-degree Hamiltonian $O$ with degree $d$. We have 
    \begin{equation} 
    \frac{1}{3} C(k,d) \paulinorm{O}{1} \leq \lVert O \rVert , 
    \end{equation}
    where $C(k,d) = \frac{\sqrt{2(k!)}}{\sqrt{d} k^{k+2.5} (2\sqrt{6} + 4\sqrt{3})^k}$ and $\lVert O \rVert$ is the spectral norm of $O$.
\end{lemma}

\renewcommand\thetheorem{\ref{Thm:Main}}
\setcounter{theorem}{\arabic{theorem}-1}
\begin{theorem}   
Given an unknown $n$-qubit Pauli channel $\mathcal{P}$, a noisy state $\cP(\sigma)$, and an $n$-qubit bounded-degree $k$-local observable $O$ where $k=\mathcal{O}(1)$. For $\epsilon, \delta > 0$, there exists an algorithm that uses $N = \mathcal{O}(\log (n^{k}/ \delta)/\epsilon^2)$ accesses to the channel to obtain a function $f(\mathcal{P}(\sigma), O)$ such that
\begin{equation}
\abs{f(\mathcal{P}(\sigma),O) - \tr(O\sigma)} \leq \epsilon 
\end{equation}
with probability at least $1-\delta$. The computation time is $\mathcal{O}(n^k \log (n^{k}))$.
\end{theorem}

\begin{proof}
The structure of the proof is as follows: we first use Hoeffding's lemma to establish a relation between $N$, the number of channel uses, and $\tilde{\epsilon}$, the accuracy which we need to estimate $\lambda_P$ to. The proof is concluded once we find the relation between $\tilde{\epsilon}$ and $\epsilon$.\\
We let $x_P = (1/3)^{\abs{P}} \lambda_P$, $D(n,k)$ be the number of n-qubit Pauli operator with weight $k$ and $T(n,k) = \sum_{l=0}^k D(n,l)$. Note that $D(n,k) = \binom{n}{k}3^k = \mathcal{O}(n^k)$.\\
We let $\tilde{\epsilon}$ to be determined later, $\tilde{\epsilon}' = \tilde{\epsilon}/3^k$. By \cref{lem:pauli evalue est} and Hoeffding's inequality,
\begin{equation} 
\mathbb{P} ( |\widehat{x}_P - x_P| \geq \tilde{\epsilon}') \leq 2e^{-N\tilde{\epsilon}'^2 / 2\times 3^{2k}} \quad \forall P,|P|\leq k.  
\end{equation}
By a union bound,
\begin{align}
    &\mathbb{P} (|\widehat{x}_P - x_P| \leq \tilde{\epsilon}'\quad \forall P,|P|\leq k)\\
    \geq &1 - \bigcup_{P} \mathbb{P}( |\widehat{x}_P - x_P| \geq \tilde{\epsilon}')\\
    \geq &1 - 2 \sum_{l=0}^k D(n,l) e^{-N\tilde{\epsilon}'^2/{2\times3^{2k}}}.
\end{align} 
Equating the last term to $1-\delta$ we get that we require 
\begin{equation} 
N = \frac{2\times3^{2k}\log (2 T(n,k) /\delta)}{\tilde{\epsilon}'^2}.
\end{equation}
Using the fact that $T(n,k) = \mathcal{O}(n^k)$, we have $N = \mathcal{O}(\log(n^{k}/\delta)/\tilde{\epsilon}'^2 )$.
By definition of $\tilde{\epsilon}'$, we also have with probability at least $1-\delta$, 
\begin{equation} \label{alpha_P error}
|\widehat{\lambda}_P - \lambda_P | \leq \tilde{\epsilon} \quad \forall P,|P|\leq k.
\end{equation}

From now on, conditioned on the above event. Using \cref{Difference},
\begin{align}
    &\lvert f(\mathcal{P}(\sigma), O) - \tr(O\sigma) \rvert\\
    = & \lvert \tr(\mathcal{P} (\overleftarrow{O}) \sigma) - \tr(O\sigma)|\\
    \leq &\sum_{P:|P|\leq k} |\lambda_P \overleftarrow{\alpha}_P - \alpha_P| \\
    = &\sum_{P:|P|\leq k} \abs{\frac{\lambda_P}{\widehat{\lambda}_P} \alpha_P - \alpha_P}.
\end{align}

We can write $\widehat{\lambda}_P = \lambda_P + \Delta_P$ where by \cref{alpha_P error} we have that $\abs{\Delta_P} < \tilde{\epsilon}$. We also assume that $\tilde{\epsilon} < \lambda_P$, which will be validated when we determine the value of $\tilde{\epsilon}$ later. Hence, 

\begin{align}
    &\sum_{P:|P|\leq k}\abs{\lambda_P \widehat{\lambda}_P^{-1}\alpha_P-\alpha_P}\\
    = &\sum_{P:|P|\leq k}\abs{\lambda_P \lambda_P^{-1} \left(1+\frac{\Delta_P}{\lambda_P}\right)^{-1} \alpha_P-\alpha_P}\\
    = &\sum_{P:|P|\leq k}\abs{\left(1-\frac{\Delta_P}{\lambda_P}+\mathcal{O}(\Delta_P^2)\right) \alpha_P - \alpha_P}.
\end{align}

Assume $\tilde{\epsilon}^2$ is negligible, and define $\lambda^{(k)}_{\min} = \min_{P:|P|\leq k} | \lambda_P |$. So $|\Delta_P/\lambda_P| \leq |\tilde{\epsilon}/\lambda^{(k)}_{\min} |$,
\begin{align}
    & \sum_{P:|P|\leq k} \abs{ \left(1 - \frac{\Delta_P}{\lambda_P} + \mathcal{O}(\Delta_P^2) \right)\alpha_P - \alpha_P } \\
    \leq & \bigg( \frac{\tilde{\epsilon}}{\lambda^{(k)}_{\min}} \bigg) \bigg( \sum_{P:|P|\leq k} |\alpha_P| \bigg)\\
    \leq & \bigg( \frac{\tilde{\epsilon}}{\lambda^{(k)}_{\min}} \bigg) \bigg( \paulinorm{O}{1} \bigg) \\
    \leq & \bigg( \frac{\tilde{\epsilon}}{\lambda^{(k)}_{\min}} \bigg)  \bigg( \frac{3\lVert O \rVert}{C(k,d)} \bigg) \\
    = & \frac{3}{\lambda^{(k)}_{\min} C(k,d)} \tilde{\epsilon},
\end{align}
where we used the bound on Pauli-1 norm from \cref{lem:Pauli norm} and $d$ is the degree of $O$.
By setting 
\begin{equation} 
\tilde{\epsilon} = \frac{\lambda^{(k)}_{\min} C(k,d)}{3}\epsilon,  
\end{equation}
we have 
\begin{equation}
\abs{f(\mathcal{P}(\sigma)) - \tr(O\sigma)} \leq \epsilon. 
\end{equation}
In conclusion, $N = \mathcal{O} (\log (n^{k}/\delta)/\tilde{\epsilon}^2) = \mathcal{O}(\log (n^k/\delta) / \epsilon^2) $. 
For total computational complexity, step 3 of the algorithm requires $\mathcal{O}(n^k\cdot N\cdot k)$ computations, step 4 and 5 of the algorithm both require $\mathcal{O}(n^k)$ computations conditioned on our assumptions that values of $\tr(P\mathcal{P}(\sigma))$ is known. So the overall computational complexity is $\mathcal{O}(Nkn^k) = \mathcal{O}(n^k \log(n^k))$.
\end{proof}
\renewcommand{\thetheorem}{\arabic{theorem}}

\section{Optimality of sample complexity of channel}\label{appendix: optimality}
First, we elaborate the task mentioned in \cref{prop:lower bound}: 
\begin{problem}\label{problem:pauli evalue est copy}
Given a Pauli channel $\mathcal{P}$. Let $\lambda^{(k)}$ denote the vector of eigenvalue of $\mathcal{P}$ for Pauli operators whose weight are at most $k$. We suppress the superscript when the value of $k$ does not change. We would like an algorithm that learns $\widehat{\lambda}$ such that $\lvert \widehat{\lambda}_P - \lambda_P \lvert \leq \epsilon$ with high probability by $N$ accesses of the channel. More precisely, for each access, the user can prepare an arbitrary input state, send it through the channel and measure the output state using arbitrary POVM. In terms of notations, for the $t$\textsuperscript{th} access, the user prepare the state $\rho_t$ and obtain output state $\mathcal{P}(\rho_t)$. Finally, the user measures using POVM $\{ \lambda^t_i \ketbra{\phi_i}{\phi_i} \}$ where $\sum_i \lambda^t_i \ketbra{\phi_i}{\phi_i} = I$. 
\end{problem}

\Cref{alg:main} indeed fits the above description and provides an upper bound for the problem. We now prove that
\begin{proposition}
To solve \cref{problem:pauli evalue est copy}, the number of channel use $N$ is at least $\Omega(\log(n^k)/\epsilon^2)$.
\end{proposition}
\begin{proof}
    We follow the idea outlined in~\cite{chen2022quantum,huang2021informationtheoretic} and~\cite{fawzi2023lower}. We first construct a set of Pauli channels labelled by $(P,s)$:
    \begin{equation}\label{eqn:lower Pauli}
        \mathcal{P}_{(P,s)}(\rho) \coloneqq \tr(\rho) \frac{I}{2^n} + 2\epsilon s \frac{\tr(P\rho)}{2^n} P ,
    \end{equation}
for $P:P\neq I, |P|\leq k$, $s\in\{\pm 1\}$. We further define $T(n,k) = \sum_{l=1}^k \binom{n}{l} 3^l$, each term in the summation is the number of Pauli operators with weight $l$, so the set defined by \cref{eqn:lower Pauli} contains $2T(n,k)$ elements. The eigenvalues of each channel is either $0$ or $\pm 2\epsilon$, so if we can learn each element of $\lambda^{(k)}$ up to accuracy $\epsilon$ with probability at least $1-\delta$, then we can determine which Pauli channel has been given. Hence we can use this set for a theoretical quantum communication task. Suppose Alice and Bob agreed an ordering of $\{ \mathcal{P}_{(P,s)}\}_{(P,s)}$. Then Alice picks a number in $[ 2T(n,k)]$ uniformly at random and sends $N$ copies of the corresponding channel to Bob. Once Bob received them, he can use the protocol that can solve \cref{problem:pauli evalue est copy} to successfully decode Alice's message with probability at least $1-\delta$. If we call the uniform distribution on $(P,s)$ $X$, and distribution of Bob's guesses $\widehat{X}$, then by Fano's inequality,
\begin{equation}
    H ( X \mid \widehat{X}) \leq \delta\log (2T(n,k)-1) + h(\delta),
\end{equation}
where $h(p)$ is the binary entropy of $\{p,1-p\}$. This then implies that
\begin{equation}\label{eqn:mutual lower bound}
    I(X : \widehat{X}) = H(X) - H(X \mid \widehat{X}) \geq \log(2T(n,k)) - h(\delta) - \delta\log(2T(n,k)-1) = \Omega(\log(n^k)).
\end{equation}
Recall that Bob's protocol involves choosing an input state $\rho$, sending it through the channel and measuring using POVM. Let $I_t$ denotes the random variable for Bob's $t$\textsuperscript{th} measurement outcome. Since $\widehat{X}$ is a function of $I_1, \ldots, I_N$, by data processing inequality and the chain rule,
\begin{align}
    I(X:\widehat{X}) &\leq I(X:I_1, \ldots, I_N)\\
    &= \sum_{t=1}^{N} I(X:I_t \mid I_1, \ldots, I_{t-1})\\
    &\leq \sum_{t=1}^{N} I(X:I_t) \label{eqn:sum mutual}.
\end{align}
Before we proceed, we make some preliminary calculations. Let $\Pr(i_t \mid (P,s))$ be the probability that the $t$\textsuperscript{th} measurement result is $i_t$ given that the channel received is $(P,s)$, we have that 
\begin{align}
    \Pr(i_t \mid (P,s)) &= \lambda_{i_t} \bra{\phi^t_{i_t}} \frac{I}{2^n} + 2\epsilon s \frac{\tr(P\rho_t)}{2^n} P \ket{\phi_{i_t}}\\
    &= \frac{\lambda_{i_t}}{2^n} [1 + 2\epsilon s \tr(P\rho_t) \bra{\phi^t_{i_t}}P \ket{\phi^t_{i_t}}].
\end{align}
From this, we can calculate $\Pr(i_t)$ by taking the marginal:
\begin{equation}
    \Pr(i_t) = \underset{P,s}{\mathbb{E}} \Pr(i_t \mid (P,s)) = \frac{\lambda_{i_t}}{2^n}.
\end{equation}
Furthermore, we calculate the squared quantities, 
\begin{equation}\label{eqn:expec sq}
    \underset{P,s}{\mathbb{E}} \Pr(i_t \mid (P,s))^2 = \frac{(\lambda^t_{i_t})^2}{2^{2n}} [1 + 4\epsilon^2 \underset{P}{\mathbb{E}} \big[\tr (P\rho_t)]^2 \bra{\phi^t_{i_t}}P\ket{\phi^t_{i_t}}^2 \big] \leq \frac{(\lambda^t_{i_t})^2}{2^{2n}} [1 + 4\epsilon^2].
\end{equation}
Now continue from \cref{eqn:sum mutual}, we want to bound $I(X:I_t)$. By definition of mutual information, we have 
\begin{align}
    I(X:I_t) &= H(I_t) - H(I_t \mid X)\\
    &= - \sum_{i_t} \Pr(i_t) \log \Pr(i_t) + \underset{P,s}{\mathbb{E}} \sum_{i_t} \Pr(i_t\mid(P,s)) \log \Pr(i_t\mid(P,s))\\
    &\leq -\sum_{i_t} \Pr(i_t) \log \Pr(i_t) + \underset{P,s}{\mathbb{E}} \sum_{i_t} \Pr(i_t \mid (P,s)) \bigg[ \log \Pr(i_t) + \frac{\Pr(i_t \mid (P,s))- \Pr(i_t)}{\Pr(i_t)}\bigg]\\
    &= \sum_{i_t} \frac{\underset{P,s}{\mathbb{E}}\Pr(i_t \mid (P,s))^2 - \Pr(i_t)^2}{\Pr(i_t)}\\
    &\leq \sum_{i_t} \frac{\lambda^t_{i_t}}{2^{n}} 4\epsilon^2\\
    &= 4 \epsilon^2,
\end{align}
where the third line follows from $\log x \leq \log y + \frac{x-y}{y}$ and taking $x= \Pr(i_t \mid (P,s))$, $y=\Pr(i_t)$. The second to last line comes from \cref{eqn:expec sq} and the last line follows from the fact that $\sum_{i_t} \lambda^t_{i_t} = 2^n$.\\
Combine this with \cref{eqn:mutual lower bound,,eqn:sum mutual}, we have that $N = \Omega(\log(n^k)/\epsilon^2)$.
\end{proof}

\section{Clifford shadows}\label{appendix: clifford shadow}

In \cref{alg:main}, we only consider collecting Pauli shadows of the channel, and a natural question to ask is what happens if we use the highly related Clifford shadow and whether it yields better sample complexity. For classical shadows of a quantum state, the complexity for using Clifford shadow is given by
\begin{proposition}[\cite{huang2020predicting}]
    Adopt a random Clifford basis primitive, where each random unitary is uniformly selected from the $n$-qubit Clifford group. Given a collection of observables $O_1,O_2,\ldots, O_M$, accuracy parameters $\epsilon,\delta \in [0,1]$,  then 
    \begin{align}
        N = \mathcal{O}\left( \frac{\log(M/\delta)}{\epsilon^2} \max_i \tr(O_i ^2) \right)
    \end{align}
    samples are required to simultaneously predict each $\tr(O_i \rho)$ up to accuracy $\epsilon$ with success probability $1-\delta$.
\end{proposition}
Compared to \cref{prop:state_pauli_shadow}, we can see that Pauli shadows have better complexity when the locality of the observables is low, but in the case of $k=n$, Clifford shadows give better complexity. Here, we show a similar result stating that Clifford shadows cannot provide any benefits for sample complexity under the assumption that the locality of the observable is $k=\mathcal{O}(1)$. Firstly, we have the following result when we change the input distribution:
\begin{proposition} \label{prop: Cliff coeff est}
    Given an $n$-qubit observable $O = \sum_P \alpha_P P$ and let $\mathcal{D}$ be a distribution of $n$-qubit states that is invariant under any $n$-qubit Clifford gate. Then for Pauli $P \in \{ I,X,Y,Z \}^{\otimes n} \backslash \{ I^{\otimes n}\}$, we have
    \begin{equation}
        \underset{\rho \sim \mathcal{D}}{\mathbb{E}} \tr(O\rho) \tr(P\rho) = \frac{2^n \mathbb{E}_{\rho \sim \mathcal{D}} \tr(\rho^2) - 1}{4^n - 1}\alpha_P.
    \end{equation}
    For  $P = I ^{\otimes n}$, we have that 
    \begin{equation}
        \underset{\rho \sim \mathcal{D}}{\mathbb{E}} \tr(O\rho) \tr(P\rho) = \alpha_I. 
    \end{equation}
\end{proposition}
\begin{proof}
    We follow the same spirit as proof of Lemma 16 of~\cite{huang2022learning}. Writing $O = \sum_Q \alpha_Q Q$, we have that 
    \begin{align}
        &\underset{\rho \sim \mathcal{D}}{\mathbb{E}} \tr(O \rho) \tr(P\rho) \\
        = &\sum_Q \alpha_Q \underset{\rho \sim \mathcal{D}}{\mathbb{E}} \tr(Q \rho) \tr(P \rho) \\
        = &\sum_Q \alpha_Q \underset{\rho \sim \mathcal{D}}{\mathbb{E}} \tr\big( (Q\otimes P) (\rho \otimes \rho) \big). \label{eqn: cont}
    \end{align}
    Since $\mathcal{D}$ is invariant under Clifford gates, we can conjugate $\rho$ by any Clifford gate and the expectation over $\rho$ does not change. In fact, we can conjugate $\rho$ by a random Clifford gate. let $U$ be a random $n$-qubit Clifford gate, we have that 
    \begin{align}
        &\sum_Q \alpha_Q \underset{\rho \sim \mathcal{D}}{\mathbb{E}} \tr \big( (Q \otimes P) (\rho \otimes \rho) \big) \\
        = &\sum_Q \alpha_Q \underset{\rho \sim \mathcal{D}}{\mathbb{E}} \, \underset{U}{\mathbb{E}} \tr\big( (Q \otimes P) (U \otimes U) (\rho \otimes \rho) (U\otimes U)^\dagger \big) \\
        = &\sum_Q \alpha_Q  \underset{\rho \sim \mathcal{D}}{\mathbb{E}} \, \underset{U}{\mathbb{E}} \tr\big( (U \otimes U)^\dagger (Q \otimes P) (U \otimes U) (\rho \otimes \rho) \big).
    \end{align}
    To evaluate $\mathbb{E}_U (U \otimes U)^\dagger (Q \otimes P) (U \otimes U)$, we use the property that Clifford gates form a 2-design which means that 
    \begin{equation}\label{eqn:exp to int}
        \underset{U}{\mathbb{E}} (U\otimes U)^\dagger (Q \otimes P) (U \otimes U) = \int_\mathcal{U} (U\otimes U) (Q\otimes P) (U\otimes U)^\dagger d\mathcal{U},
    \end{equation}
    where $\mathcal{U}$ is the Haar measure on the unitary group of $n$-qubits. To evaluate the integral, we can use standard result:
    \begin{proposition}
        Let $\mathcal{U}(d)$ be the unitary group of $\mathbb{C}^d$, then for any $X$ which is a linear operator acting on $\mathbb{C}^d \otimes \mathbb{C}^d$, we have that 
        \begin{equation}\label{eqn:Haar integral}
            \int_{\mathcal{U}(d)} (U \otimes U) X (U \otimes U)^\dagger d\mathcal{U} = \bigg( \frac{\tr(X)}{d^2 -1} - \frac{\tr(XF)}{d(d^2-1)} \bigg) id_{d^2} - \bigg( \frac{\tr(X)}{d(d^2-1)} - \frac{\tr(XF)}{d^2 - 1} \bigg) F,
        \end{equation}
        where $F$ is the swap operator.
    \end{proposition}
\Cref{eqn:Haar integral} allows us to evaluate the integral in \cref{eqn:exp to int} to obtain
\begin{equation}
    \int_\mathcal{U} (U \otimes U) (Q \otimes P) (U \otimes U)^\dagger d\mathcal{U} = \begin{cases}
        \frac{2^n}{4^n-1}F - \frac{1}{4^n-1}id_{d^2} & \text{if } Q = P \neq I^{\otimes n}, \\
        id_{d^2} & \text{if } Q = P = I^{\otimes n},\\
        0 & \text{if } Q\neq P,
    \end{cases}
\end{equation}
where $id$ denotes the identity map. Substitute this into \cref{eqn: cont} and we obtain the desired result. 
\end{proof}

This means that our empirical estimation is exponentially small in $n$, which gives an exponential factor to our sample complexity. 

Secondly, if we reconstruct $\mathcal{P}(\rho_i)$ using Clifford shadow, we replace $\bigotimes_j (3\ketbra{t_{ij}}{t_{ij}}-I)$ with $(2^n+1)U^\dagger \ketbra{b}{b} U - I$ during calculation of $\widehat{x}_P$. Our empirical estimation still equals to $\lambda_P$ in expectation but the range of value it takes extends to $[- (2^n+1), 2^n+1]$, which also gives an exponential factor to our sample complexity following Hoeffding's inequality.
\begin{remark}
For $k=n$, using Pauli shadows would give a $3^n$ factor to our sampling complexity. This is worse than both factors above which are about $2^n$. This agrees with the intuition that random Pauli is preferred only for local observables over random Clifford.
\end{remark}

\section{Extension to other channels} \label{apx: extension}
For a general channel $\mathcal{E}$, the action of the adjoint map of the channel $\mathcal{E}^\dagger$ is also a linear map. In fact, it is a completely positive unital map, and it can be written as a matrix. We denote $M$ be the matrix in Pauli basis. For Pauli channel, $M$ is a diagonal matrix. It is formally defined as $M_{PQ} = \frac{1}{2^n} \tr (P \mathcal{E}^\dagger (Q))$. This is very similar to the Pauli transfer matrix of a quantum channel.
\begin{definition}
    The \textbf{Pauli transfer matrix} $T$ of a quantum channel $\mathcal{N}$ is defined as $T_{PQ} = \frac{1}{2^n} \tr (P \mathcal{N}(Q))$.
\end{definition}

\begin{proposition} \label{prop:adjoint_ptm_relation}
Given a quantum channel $\mathcal{N}$, and let $T$ and $M$ be defined as above, then $M = T^\intercal$, the transpose of $T$ taken with respect to the Pauli basis.
\end{proposition}
\begin{proof}
    $M_{PQ} = \frac{1}{2^n} \tr(P \mathcal{N}^\dagger (Q)) = \frac{1}{2^n} \tr(\mathcal{N}(P) Q) = T_{QP}$ where the second equality comes from definition of adjoint map.
\end{proof}

\begin{proposition}
    Let $A$ be an $n \times n$ upper block triangular matrix with block size $\{a_1, a_2, \ldots, a_m \}$. Let $s_i = \sum_{j=1}^{m} a_j$, so $s_m=n$ and the $s_i \times s_i$ submatrices of $A$ are all upper block triangular. Denote each submatrix as $A_i$ and assume each $A_i$ is invertible. Given $b$ with $b_k = 0$ $\forall k>s_i$ for some  $i$, then there exists $x$ with $x_k =0$ $\forall k>s_i$ such that $Ax = b$. To be more specific, $x$ is given by
    $\begin{pmatrix}
    A_i^{-1} b^{(i)} \\
    0
    \end{pmatrix}$ where $b^{(i)}$ denotes the subvector $(b_1, b_2, \ldots, b_{s_i})^\intercal$ of $b$.
\end{proposition}
\begin{proof}
    Direct verification.
\end{proof}

The above observation ensures that if $M$ is upper block triangular with block $i$ being the set of $i$-local Pauli operators, then to find the inverse of a $k$-local observable under the adjoint map, one can achieve by finding $s_k \times s_k$ submatrix of $M$ and the inverse is also a $k$-local observable. 
From now on, when we refer to the term `upper block triangular', we mean upper block triangular with block $i$ being the set of $i$-local Pauli operators. Pauli channel being a diagonal channel automatically satisfies the upper block triangular condition. Another class of channels that meets this criterion is product channel. Using \cref{prop:adjoint_ptm_relation}, this condition also translates to the Pauli transfer matrix being lower block triangular.
\begin{proposition}
    If a noise channel factorises, i.e. $\mathcal{N} = \bigotimes_i \mathcal{N}_i $, then the matrix of $\mathcal{N}^\dagger$ in Pauli basis is upper block triangular. 
\end{proposition}
\begin{proof}
    We notice that a necessary and sufficient condition for an adjoint map to be upper triangular is that it preserves the weight of a Pauli operator. That is, the Pauli decomposition of $\mathcal{N}^\dagger (P) $ does not contain terms whose weight is greater than $P$, i.e. 
    \begin{equation}
        \mathcal{N}^\dagger(P) = \sum_{Q : \abs{Q}\leq \abs{P}} \alpha_Q Q . 
    \end{equation}
    We say that $\mathcal{N}^\dagger$ is weight contracting. If a noise channel $\mathcal{N}$ factorises, then its adjoint map also factorises, i.e. $\mathcal{N}^\dagger = \bigotimes_i \mathcal{N}^\dagger_i$. Given a Pauli operator $P = \bigotimes_i P_i$, $\mathcal{N}^\dagger (P) = \bigotimes_i \mathcal{N}^\dagger_i (P_i)$. If $P_i = I$, $\mathcal{N}^\dagger_i (P_i) = I$ since each $\mathcal{N}^\dagger_i$ is an adjoint map itself and hence is unital. Therefore the weight of $\mathcal{N}^\dagger (P)$ does not increases, hence $\mathcal{N}^\dagger$ is weight contracting.
\end{proof}

To extend the previous method, instead of calculating the diagonal elements, we need to calculate the entire upper triangular block matrix. Then taking the inverse of the resultant matrix gives us estimation for $\overleftarrow{O}$. Hence we can generalise the algorithm and theorem in \cref{sec:alg,sec:analysis} to extend the set of channels which we can recover information for. Before doing that, we will show how to estimate the non-diagonal element of $M$.

\begin{lemma}[Lemma 16 of Ref.~\cite{huang2022learning}] \label{lem:estimate coeff}
    Given an n-qubit observable $O = \sum_P \alpha_P P$ and let $\mathcal{D}^0$ be the uniform distribution of product state of Pauli eigenstates. Then for Pauli $P \in \{ I,X,Y,Z \}^{\otimes n}$, we have
    \begin{equation}
    \underset{\rho \sim \mathcal{D}^0}{\mathbb{E}} \tr(O\rho) \tr(P\rho) = \bigg( \frac{1}{3} \bigg)^{|P|} \alpha_P        .
    \end{equation}
\end{lemma}
If we take $O = \mathcal{E}^\dagger (Q)$ where $Q$ is a Pauli operator, then the above lemma allows us to obtain the coefficient of $\mathcal{E}^\dagger (Q)$, which corresponds to column $Q$ of $M$. More specifically,
\begin{equation}
    \bigg( \frac{1}{3} \bigg)^{|P|} M_{PQ} = \underset{\rho \sim \mathcal{D}^0}{\mathbb{E}} \tr (\mathcal{E}^\dagger (Q) \rho) \tr(P \rho) = \underset{\rho \sim \mathcal{D}^0}{\mathbb{E}} \tr (Q \mathcal{E}(\rho)) \tr(P \rho).
\end{equation}
Another difference is that instead of dividing each coefficient of our observable by the corresponding estimated eigenvalue, we need to apply the inverse of our estimated $M$ to the vector of coefficients. This involves inverting a square matrix whose dimension of $\mathcal{O}(n^k)$ which has computational complexity of $\mathcal{O}(n^{3k})$. The updated algorithm and its complexity are as follows:
\begin{algorithm}[H]
\caption{Modified algorithm for weight contracting channel}
\label{alg:update}
\begin{algorithmic}[1]
    \Require Unknown weight contracting channel $\mathcal{E}$, $k$-local observable $O=\sum_P \alpha_P P$, copies of unknown noisy state $\mathcal{E}(\sigma)$
    \Ensure Estimation of $\tr (O\sigma)$
    \State Prepare $N$ random product Pauli eigenstates $\{ \rho_i = \bigotimes_{j=1}^n \ketbra{s_{ij}}{s_{ij}} \}_{i=1}^N$. Send them through the unknown channel $\mathcal{E}$ and make random Pauli measurements on each qubit and obtain data $\{ \bigotimes_{j=1}^n \ketbra{t_{ij}}{t_{ij}} \}_{i=1}^N$.
    \State For each n-qubit Pauli operator $Q$ with $\lvert Q \rvert\leq k$, compute $\widehat{\lambda}_P(Q)$ described below for P such that |P|$\leq$ |Q|. \Comment{The computation time is $\mathcal{O}(n^{2k} \cdot N \cdot k)$}
    \State Construct matrix $\tilde{M}$ such that $\tilde{M}_{PQ} = \widehat{\lambda}_P(Q)$. $\tilde{M}$ is an upper block triangular matrix.
    \State For $O = \sum_{P:|P|\leq k} \alpha_P P$. Let $\overleftarrow{\alpha} = \tilde{M}^{-1} \alpha$. \Comment{The computation time is $\mathcal{O}(n^{3k})$}
    \State To make prediction, given $\mathcal{E}(\sigma)$, we compute $f(\mathcal{E}(\sigma)) = \sum_{P : |P|\leq k} \overleftarrow{\alpha}_P \tr (P \mathcal{E}(\sigma))$. \Comment{The computation time is $\mathcal{O}(n^k)$}
\end{algorithmic}
\end{algorithm}

We compute $\widehat{\alpha}_P(Q)$ as follows:

We first let
\begin{align}
\widehat{x}_P(Q) &= \frac{1}{N} \sum_{i=1}^{N} \tr \big(Q \bigotimes_{j=1}^n (3 \ketbra{t_{ij}}{t_{ij}} - I) \big) \tr (P \rho_i) \\
  &= \frac{1}{N} \sum_{i=1}^N \prod_{j=1}^n \tr(Q_j (3 \ketbra{t_{ij}}{t_{ij}} - I)\tr(P_j \ketbra{s_{ij}}{s_{ij}}),
\end{align}
then
\begin{equation}
    \widehat{\lambda}_P(Q) = 3^{|P|} \widehat{x}_P(Q).
\end{equation}
Again, we have the special case of $\lambda_P(I) = \delta_{PI}$ from $\mathcal{E}^\dagger$ being a unital map, so there is no need to perform estimation.

\renewcommand\theproposition{\ref{prop:extension}}
\setcounter{proposition}{\arabic{proposition}-1}
\begin{proposition}
Given an unknown $n$-qubit weight contracting channel $\mathcal{E}$, a noisy state $\cE(\sigma)$, and an $n$-qubit bounded-degree $k$-local observable O with $\norm{O}_\infty = 1$. For $\eps, \delta > 0$, there exists an algorithm that uses $N = \mathcal{O}(n^{2k}\log (n^{2k}/ \delta)/\epsilon^2)$ access to the channel to obtain a function $f$ such that
\begin{equation}
\abs{ f(\mathcal{E}(\sigma)) - \tr(O\sigma)} \leq \epsilon 
\end{equation}
with probability at least $1-\delta$. The computation time is $\mathcal{O}(n^{4k}\log(n^{2k}))$.
\end{proposition}
\renewcommand{\theproposition}{S\arabic{proposition}}

\begin{proof}
The general structure of the proof is the same as before, but we have more coefficients. For ease of comparison to previous proof, we let $\lambda_P(Q) = M_{PQ} = \frac{1}{2^n} \tr(P \mathcal{N}^\dagger (Q))$, $x_P(Q) = (1/3)^{|P|} \lambda_P(Q)$. $D(n,k)$ be the number of $n$-qubit Pauli operator with weight $k$ and $T(n,k) = \sum_{l=0}^k D(n,l)$.

We let $\tilde{\epsilon}$ to be determined later, $\tilde{\epsilon}' = \tilde{\epsilon}/3^k$. By Hoeffding's inequality, we have

\begin{equation}
     \mathbb{P} (\abs{\widehat{x}_P(Q) - x_P(Q)} > \tilde{\epsilon}') \leq 2 e^{-N\tilde{\epsilon}'^2 / 2\times 3^{2k}} \quad \forall P,Q,|P|\leq|Q|\leq k.
\end{equation}
By a union bound,
\begin{align}
    &\mathbb{P} (|\widehat{x}_P(Q) - x_P(Q)| \leq \tilde{\epsilon}'\quad \forall P,Q,|P|\leq|Q|\leq k)\\
    \geq &1 - \bigcup_{P,Q} \mathbb{P}( |\widehat{x}_P(Q) - x_P(Q)| > \tilde{\epsilon})\\
    \geq &1 - 2 \sum_{l=0}^k D(n,l) \sum_{m=0}^{l} D(n,m)  e^{-N\tilde{\epsilon}'^2 / 2\times 3^{2k}} \\
    = &1 - 2 \sum_{l=0}^k D(n,l)T(n,l)  e^{-N\tilde{\epsilon}'^2 / 2\times 3^{2k}}. \\
\end{align} 
Equating the last term to $1-\delta$, we get that we require
\begin{equation}
    N = \frac{2\times 3^{2k}\log(2\sum_l D(n,l)T(n,l)/\delta)}{\tilde{\epsilon}'^2}.
\end{equation}
We note that $\sum_l D(n,l)T(n,l) < T(n,k)^2 = \mathcal{O}(n^{2k})$, so it holds that $N = \mathcal{O} (\log (n^{2k}/\delta)/ \tilde{\epsilon}'^2)$. By definition of $\tilde{\epsilon}'$, this $N$ also ensures that with probability at least $1-\delta$, 
\begin{equation}
|\widehat{\lambda}_P (Q) - \lambda_P (Q) | \leq \tilde{\epsilon} \quad \forall P,Q,|P|<|Q|\leq k.
\end{equation}

Conditioned on this, following the same error analysis, we have that 
    \begin{align}
        &  \abs{ \tr (\mathcal{E}^\dagger (\overleftarrow{O}) \sigma) - \tr (O \sigma) } \\
        = &\sum_{P:\abs{P}\leq k}  \abs{(M \overleftarrow{\alpha})_P - \alpha_P} \\
        \leq & \norm{ M \overleftarrow{\alpha} - \alpha }_1\\
        = & \norm{ M (\tilde{M})^{-1} \alpha - \alpha}_1.
    \end{align}
We can write $\tilde{M} = M + \Delta$, where by construction, each element of $\Delta$ has modulus at most $\tilde{\epsilon}$. So $\lVert \Delta \rVert_1 \leq T(n,k) \tilde{\epsilon}$. We also use the fact that $(M+\Delta)^{-1} = M^{-1} - M^{-1}\Delta M^{-1} + \mathcal{O}( \lVert \Delta \rVert^2)$. Assuming $\mathcal{O} (\lVert \Delta \rVert^2)$ is negligible, 
\begin{align}
    &\lVert M (\tilde{M})^{-1} \alpha - \alpha \rVert_1\\
    \leq &\lVert M \rVert_1 \lVert (M^{-1} - M^{-1} \Delta M^{-1}) \alpha - M^{-1} \alpha \rVert_1\\
    = &\lVert M \rVert_1 \lVert - M^{-1} \Delta M^{-1} \alpha \rVert_1\\
    \leq & \lVert M\rVert_1 \lVert M^{-1} \rVert_1 \lVert M^{-1} \rVert_1 \lVert \Delta \rVert_1 \lVert \alpha \rVert_1  \\
    \leq &B(M)  T(n,k) \tilde{\epsilon} \bigg( \frac{3\lVert O \rVert}{C(k,d)} \bigg),
\end{align}
where we define $B(M) = \lVert M \rVert_1 \lVert M^{-1} \rVert^2_1 $ and note that $\lVert \alpha \rVert_1 = \paulinorm{O}{1}$. Again, $d$ is the degree of $O$.\\
By setting
\begin{equation}
    \tilde{\epsilon} = \frac{C(k,d)}{3 B(M) T(n,k)} \epsilon,
\end{equation} 
we have 
\begin{equation}
    \lvert f(\mathcal{E}(\sigma)) - \tr(O\sigma) \rvert \leq \epsilon. 
\end{equation} 
In conclusion, $N = \mathcal{O}(\log(n^{2k}/\delta)/\tilde{\epsilon}^2) = \mathcal{O}(n^{2k}\log(n^{2k}/\delta)/ \epsilon^2)$, and the corresponding total computation time is $\mathcal{O}(n^{4k}\log(n^{2k}))$.
\end{proof}

\section{Application in error mitigation for Clifford circuits} \label{apx:Clifford application}
In error mitigation, the most widely used noise model is GTM, stands for gate-independent, time-stationary Markovian noise. In the case of a Clifford circuit made of H, S and CNOT gates, it means that every time a gate is used, it is followed by the same noise channel. We consider the slightly stronger case where the noise is dependent on gate, i.e. we assume that the same channel follows the same type of gates and each noise channel is a Pauli channel. Hence given a circuit $\mathcal{C}$ and arbitrary input state $\sigma$, the ideal output state is $\mathcal{C}(\sigma)$ but instead, we get the noisy state $\tilde{\mathcal{C}}(\sigma)$ where the noise is described as above. And subsequently, when measured using an observable $O$, this would produce an error. The goal of error mitigation is recovering the true value of $\tr (O \mathcal{C}(\sigma))$ when we only have access to the noisy states $\tilde{\mathcal{C}}(\sigma)$. We show how we can adapt our method to perform this task when the observable is $k$-local. The main idea is still first by accessing the gates with random Pauli eigenstates to learn the eigenvalues of the Pauli channels, then using properties of Clifford gates being stabilizers of the Pauli groups, we show how we still only need to scale coefficients of the observables to make an estimation, similar to what we did in \cref{Thm:Main}. 
In the learning part, we once again send random Pauli eigenstates to a single gate and then measure in random Pauli basis. The circuit can be depicted as:
\[\begin{array}{c}
\Qcircuit @C=1em @R=1em{
    & \lstick{\rho_i} & \gate{U} & \gate{\mathcal{P}_U} & \qw}
\end{array}\]
where $U \in \{H,S,CNOT \}$, $\mathcal{P}_U$ is the corresponding Pauli noise and $\rho_i$ is a Pauli eigenstate. The effect of single qubit $U$ acting on $\rho_i$ simply changes it to a different Pauli eigenstate. This brings it back to the same setting as \cref{Thm:Main}, and we can estimate eigenvalues of $\mathcal{P}_U$ following \cref{alg:main} except in \cref{eqn:xp}, we need to change it to
\begin{equation}
    \widehat{x}_P = \frac{1}{N} \sum_{i=1}^N \tr( P ( 3\ketbra{t_{i}}{t_{i}} - I)) \tr (P (U\rho_i U^\dagger)).
\end{equation}
For $\cnot$ gate, the matrix of noisy $\cnot$ gate is a monomial matrix. If we label entries by Pauli operators, then $(\pi(P),P)$ entry takes the value of $\lambda_P$ where $\pi(P)$ is the resultant Pauli operator by conjugating $P$ by $\cnot$. In this case, by modifying \cref{lem:estimate coeff}, we can estimate $\lambda_P$ as follows:
\begin{align}
    \widehat{x}_P &= \frac{1}{N}  \sum_{i=1}^N \tr (P \bigotimes_{j=1}^2 (3 \ketbra{t_{ij}}{t_{ij}} - I)) \tr (\pi(P) \rho_i), \\
    \widehat{\lambda}_P &= 3^{|\pi(P)|} \widehat{x}_P.
\end{align}

Next, given an $n$-qubit Clifford circuit $\mathcal{C}$, it can be written as 
\begin{equation}
    \mathcal{C} = \mathcal{U}_d \circ \mathcal{U}_{d-1} \circ \cdots \circ \mathcal{U}_1 ,
\end{equation} 
where $\mathcal{U}_i (\rho) = U_i \rho U_i^\dagger$, $U_i \in \{ H,S,CNOT \}$. The corresponding noisy circuit $\tilde{\mathcal{C}}$ is given by 
\begin{equation}
    \tilde{\mathcal{C}} = \mathcal{P}_{\mathcal{U}_d} \circ \mathcal{U}_d \circ \cdots \circ \mathcal{P}_{\mathcal{U}_1} \circ \mathcal{U}_1. 
\end{equation} 
Furthermore, since each $U_i$ is a Clifford gate, it permutes the $n$-qubit Pauli group. We denote the permutation as $\pi_i$, i.e. $\pi_i(P) = U_i^\dagger P U_i$. The eigenvalue of $P$ under $\mathcal{P}_{\mathcal{U}_i}$ is denoted as $\lambda_{i,P}$. We note that after the learning part, both $\pi_i$ and $\lambda_{i,P}$ can be efficiently computed classically. (since we can view each gate (and noise) as gate (and noise) on $n$-qubits by taking tensor product with identity.)

Finally, given $k$-local observable $O = \sum \alpha_P P$ and copies of $\tilde{\mathcal{C}}(\sigma)$, we would like to find estimation for $\tr(O \mathcal{C}(\sigma))$. The idea would be once again find $\overleftarrow{O} = \sum_P \overleftarrow{\alpha}_P P$ so that $\tr (\overleftarrow{O} \tilde{\mathcal{C}}(\sigma)) = \tr(O \mathcal{C}(\sigma))$. Take $d=2$ as an example and it can be extended to arbitrary $d$ by induction. Let 
\begin{equation}
    \tilde{\mathcal{C}} = \mathcal{P}_{\mathcal{U}_2} \circ \mathcal{U}_2 \circ \mathcal{P}_{\mathcal{U}_1} \circ \mathcal{U}_1,
\end{equation}  
then $\tr (\overleftarrow{O} \tilde{\mathcal{C}}(\sigma))$ can be written as 
\begin{equation}
    \tr(\overleftarrow{O} \tilde{\mathcal{C}}(\sigma)) = \sum_P \overleftarrow{\alpha}_P \tr( P \: \mathcal{P}_{\mathcal{U}_2} \circ \mathcal{U}_2 \circ \mathcal{P}_{\mathcal{U}_1} \circ \mathcal{U}_1 (\sigma)) .
\end{equation}
Considering the trace term, we can simplify it as follows:
\begin{align}
    \tr( P\mathcal{P}_{\mathcal{U}_2} \circ \mathcal{U}_2 \circ \mathcal{P}_{\mathcal{U}_1} \circ \mathcal{U}_1 (\sigma)) &= \tr( \mathcal{P}_{\mathcal{U}_2}(P) \: \mathcal{U}_2 \circ \mathcal{P}_{\mathcal{U}_1} \circ \mathcal{U}_1 (\sigma)) \\
    &= \lambda_{2,P} \tr(P \: \mathcal{U}_2 \circ \mathcal{P}_{\mathcal{U}_1} \circ \mathcal{U}_1 (\sigma)) \\
    &=  \lambda_{2,P} \tr( \mathcal{U}^\dagger_2 (P) \: \mathcal{P}_{\mathcal{U}_1} \circ \mathcal{U}_1 (\sigma)) \\
    &= \lambda_{2,P} \tr( \mathcal{P}_{\mathcal{U}_1}(\pi_2(P)) \: \mathcal{U}_1 (\sigma)) \\
    &= \lambda_{2,P} \lambda_{1, \pi_2(P)} \tr (\pi_2(P) \: \mathcal{U}_1 (\sigma)) \\
    &= \lambda_{2,P} \lambda_{1, \pi_2(P)} \tr (\mathcal{U}^\dagger_2 (P) \: \mathcal{U}_1 (\sigma)) \\
    &= \lambda_{2,P} \lambda_{1, \pi_2(P)} \tr (P \: \mathcal{U}_2 \circ \mathcal{U}_1 (\sigma))\\
    &= \lambda_{2,P} \lambda_{1, \pi_2(P)} \tr (P \: \mathcal{U}_2 \circ \mathcal{U}_1 (\sigma)) \\
    &= \lambda_{2,P} \lambda_{1, \pi_2(P)} \tr (P \: \mathcal{C} (\sigma)).
\end{align}
Hence by setting
\begin{equation}
    \overleftarrow{\alpha}_P = \frac{\alpha_P}{\lambda_{2,P} \lambda_{1, \pi_2(P)}},
\end{equation}
we have that 
\begin{align}
    \tr(\overleftarrow{O} \tilde{\mathcal{C}}(\sigma)) &= \sum_P \overleftarrow{\alpha}_P \tr( P\mathcal{P}_{\mathcal{U}_2} \circ \mathcal{U}_2 \circ \mathcal{P}_{\mathcal{U}_1} \circ \mathcal{U}_1 (\sigma)) \\    
    &= \sum_P \alpha_P \tr(P \: \mathcal{C}(\sigma)) \\
    &= \tr(P \mathcal{C}(\sigma)).
\end{align}

\textit{SPAM error.}
Random Pauli input states are prepared by applying a single gate to the zero state. Suppose these gates also have Pauli noise (so there is state preparation error), how does this affect our estimation? The effect of Pauli noise on Pauli eigenstates is flipping it to the eigenstate with the negative eigenvalue with probability related to the parameters of the channel. For example, $X\ket{0} = \ket{1}, \ Y\ket{0} = \ket{1}, \ Z\ket{0}=\ket{0}, \ I\ket{0}=\ket{0}$. So if Pauli channel has probability $\{ p_I, p_X, p_Y, p_Z \}$, then it will keep $\ket{0}$ unchanged with probability $p_I + p_Z$ and flip it to $\ket{1}$ with probability $p_X + p_Y$. This applies to the other 5 states as well but with different probability. Now, if we consider the case where the noise is a depolarizing channel with probability of $p$, then the probability of flipping for every Pauli eigenstate is the same. Let $\rho$ denotes a Pauli eigenstate, $\rho'$ be the state orthogonal to it, $p_0$ be the probability that the depolarizing channel preserves $\rho$ and $p_1$ be the probability of flipping to $\rho'$. Now, substitute this into the left hand side of \cref{lem:pauli evalue est}, we have that 
\begin{align}
    &\underset{\rho \sim \mathcal{D}^0}{\mathbb{E}} \tr(P \mathcal{P}(p_0 \rho + p_1 \rho')) \tr (P (p_0\rho + p_1 \rho'))\\
    = &\underset{\rho \sim \mathcal{D}^0}{\mathbb{E}} p_0^2 \tr (P\mathcal{P}(\rho))\tr (P\rho) + p_0p_1 \tr (P\mathcal{P}(\rho))\tr (P\rho') + p_0p_1 \tr (P\mathcal{P}(\rho'))\tr (P\rho) + p_1^2 \tr (P\mathcal{P}(\rho'))\tr (P\rho') \\
    = &\underset{\rho \sim \mathcal{D}^0}{\mathbb{E}} p_0 (p_0 - p_1) \tr (P\mathcal{P}(\rho))\tr (P\rho) - p_1(p_0 -p_1)\tr (P\mathcal{P}(\rho'))\tr (P\rho') \\
    = &\underset{\rho \sim \mathcal{D}^0}{\mathbb{E}} p_0 (p_0 - p_1) \tr (P\mathcal{P}(\rho))\tr (P\rho) - p_1(p_0 -p_1)\tr (P\mathcal{P}(\rho))\tr (P\rho) \\
    = &\underset{\rho \sim \mathcal{D}^0}{\mathbb{E}} (p_0-p_1)^2 \tr (P\mathcal{P}(\rho))\tr (P\rho) \\
    = &(p_0 - p_1)^2 \frac{1}{3^{|P|}} \lambda_P,
\end{align}
where in the second line we use linearity of trace and $\mathcal{P}$, in the third line we use $\tr(P\rho) = - \tr(P \rho')$, in the fourth line we use linearity of expectation and the fact that $\rho'$ has the same distribution as $\rho$. Finally, in the last line we used the result of \cref{lem:pauli evalue est}. This means that when we are estimating using \cref{eqn:xp} to \cref{eqn:lambdap}, we would pick up an extra factor of $(p_0 - p_1)^2$ to all coefficients (expect the one for identity where we just set it to 1). This factor can be obtained if we perform the whole procedure with only state preparation and measurement, and use it to estimate known quantity of traceless observable, for example estimate the expectation value of $\ketbra{+}{+}$ with respect to $X$. Their ratio is our estimation for the prefactor, and we can divide it out from our estimate to remove the effect of state preparation error.
\end{document}